\documentclass[12pt]{article}
\usepackage{amsmath,amssymb,amsthm,graphicx}

\usepackage[pdfpagemode=UseOutlines,
bookmarks=true, bookmarksopenlevel=3, bookmarksopen=true,
bookmarksnumbered=true, unicode=true,
pdffitwindow=true,pdfpagelayout=SinglePage,pdfhighlight=/P,
colorlinks=true,linkcolor=blue,citecolor=blue,urlcolor=blue]{hyperref}

\usepackage[left=2cm, right=2cm, top=2cm, bottom=2cm]{geometry}
\allowdisplaybreaks[2]

\usepackage[width=0.95\textwidth,indention=0mm,justification=centerlast,font=small,labelfont=bf,labelsep=period]{caption}

\numberwithin{equation}{section}
\theoremstyle{plain}
\newtheorem{proposition}{Proposition}
\theoremstyle{remark}
\newtheorem{remark}{Remark}

\title{Self-similar asymptotics in the decay problem for the Volterra lattice with zero boundary condition}

\author{V.E.\:Adler\thanks{L.D.\:Landau Institute for Theoretical Physics, Chernogolovka, Russian Federation.},\quad 
B.I.\:Suleimanov\thanks{Institute of Mathematics, Ufa Federal Research Centre, Russian Academy of Sciences, 112, Chernyshevsky Street, Ufa 450008, Russian Federation.}} 

\date{June 11, 2026}

\begin{document}

\maketitle
\begin{abstract}
The article is devoted to the problem of decay of initial stationary state for the Volterra lattice with zero boundary condition. We show that this process is asymptotically self-similar and calculate the propagation velocity of the decay wave, the leading terms of the asymptotics and corrections, in the main and transition sectors of the wave.\medskip

Key words: Volterra lattice, initial-boundary value problem, decay wave, shock wave, asymptotics, Painlev\'e equation, Hastings--McLeod solution.
\end{abstract}

\section{Introduction}

For the propagation of waves in nonlinear media, the processes of decay of initial discontinuity and wave breaking are of great importance. These processes were first described by Gurevich and Pitaevskii \cite{Gurevich_Pitaevskii_1973, Gurevich_Pitaevskii_1974} within the framework of the Whitham averaging method, for model examples such as the Korteweg--de Vries equation and the nonlinear Schr\"odinger equation. These problems have been studied in a lot of works. In particular, in addition to the Whitham method, the inverse scattering method was developed in \cite{Hruslov_1976, Hruslov_Kotlyarov_1994}, the well-posedness of various formulations of the Cauchy problem with nonstandard boundary conditions was studied in \cite{Cohen_1984, Kappeler_1986, Pokhozhaev_2012, Rybkin_2017}, asymptotic formulas were refined and generalized for different types of initial data and for other nonlinear equations in \cite{Venakides_1986, Bikbaev_1989, Novokshenov_2005, Egorova_Gladka_Kotlyarov_Teschl_2013}. A detailed bibliography, history of the problem and an overview of its further development can be found in \cite{Kamchatnov_2021, Kamchatnov_2025, Kamchatnov_2026}. We also mention that many works were devoted to construction of exact solutions expressed in terms of special solutions of higher analogues of the Painlev\'e equations, for instance \cite{Suleimanov_1994, Kudashev_Suleimanov_1996, Garifullin_Suleimanov_Tarkhanov_2009, Garifullin_Suleimanov_2010, Dubrovin_2006, Claeys_Vanlessen_2007, Claeys_Grava_2010, Claeys_Grava_2011, Adler_2020} if we limit ourselves to just the example of the KdV equation. Similar problems have also been studied for equations with discrete variables, also based on the Whitham averaging and inverse scattering methods. The simplest and at the same time fundamental model here is the Volterra lattice (VL)
\begin{equation}\label{VL}
 \dot u_n = u_n(u_{n+1}-u_{n-1}),
\end{equation}
as well as the modified Volterra and Toda lattices related with it by difference substitutions. Recall that (\ref{VL}) defines a discretization of the differential-integral kinetic equation for one-dimensional Langmuir waves in plasma \cite{Zakharov_Musher_Rubenchik_1974, Manakov_1974}. In another well-known interpretation, equations (\ref{VL}) describe an idealized food chain of biological species in which species $n$ has a population of size $u_n(t)$, feeds on species $n+1$, and serves as food for species $n-1$.

The decay of initial discontinuity for the lattice (\ref{VL}) is formulated as a Cauchy problem for all $n\in{\mathbb Z}$ with initial data tending to different nonzero constants as $n\to\pm\infty$, possibly with different limit values for even and odd $n$. The asymptotic behavior of such solutions was studied in \cite{Vereshchagin_1997}, where it was shown that it is described by modulated cnoidal waves. Similar problems for the Toda lattice were studied in  \cite{Guseinov_Khanmamedov_1999, Monvel_Egorova_2000, Egorova_Michor_Teschl_2009}.

The shock wave problem for (\ref{VL}) is formulated on a half-line as follows. Let the system be at rest, and at time $t=0$ one variable changes and is subsequently fixed: $u_0(t)=\operatorname{const}$. This results in two non-interacting subsystems for $n>0$ and $n<0$, and the task is to describe their evolution. We present some numerical results for this problem in Section \ref{s:shock}, where we compare it with a closely related problem on shock waves in the Toda lattice, studied in \cite{Holian_Straub_1978, Holian_Flaschka_McLaughlin_1981, Venakides_Deift_Oba_1991, Kamvissis_1993, Deift_Kamvissis_Kriecherbauer_Zhou_1996, Deift_McLaughlin_1998}. However, the main part of our paper is devoted to the simpler limiting case with zero boundary $u_0(t)=0$, which, as is easy to see, amounts to solving the initial-boundary value problem
\begin{gather}
\label{VL+}
 \dot u_n=u_n(u_{n+1}-u_{n-1}),\qquad n=1,2,\dotsc,\qquad t\in\mathbb{R},\\
\label{VL0}
 u_0(t)=0,\qquad u_1(0)=u_3(0)=\dotsb=\alpha,\qquad u_2(0)=u_4(0)=\dotsb=1/\alpha,
\end{gather}
where $\alpha>0$ is a constant parameter. Indeed, given the invariance under reflection $u_n(t)\leftrightarrow u_{-n}(-t)$, we can restrict ourselves to considering the dynamics only for $n>0$, but for all $t$. A more general equilibrium position $u_{2k-1}=\alpha_1>0$, $u_{2k}=\alpha_2>0$ is reduced to the initial data (\ref{VL0}) by scaling; this normalization will be convenient for what follows. It turns out that under the zero boundary, shock waves are not formed and the solution enters a self-similar regime, the description of which is our main goal.

The problem (\ref{VL+}), (\ref{VL0}) has already been studied in papers \cite{Kulaev_Shabat_2019, Adler_Shabat_2018, Adler_Shabat_2019, Adler_2024}, where it was shown to be exactly solvable, in a sense. In particular, it was proved that $u_n$ for fixed $n$ satisfies Painlev\'e type ODE in $t$ (equivalent to P-V), and it was shown that $u_1$ is expressible in terms of the confluent hypergeometric function ${}_1F_1$. For the remaining $u_n$, there exist explicit formulas in terms of Hankel type determinants of this function. However, the size of determinants grows as $n/2$, which makes these formulas ineffective even for very small $n$. They do not give an idea of how the solution behaves as a whole, that is, what the solution profile looks like as a function of $n$ and how it changes depending on $t$.

The numerical solution is shown in Fig.~\ref{fig:solution}, where we have chosen $\alpha=1$ for sake of simplicity (graphs for $\alpha\ne1$ are given in the next section). The observed behavior of the system is intuitively clear from the ecological interpretation given above. At time $t=0$, species $n=0$ is completely extinct due to an external influence, which means that species $n=1$ loses enemies but has food, hence $u_1$ begins to increase with increasing $t$. This means an increase in the enemy population for species $n=2$, hence $u_2$ begins to decrease. In turn, this leads to an increase of $u_3$, a decrease of $u_4$, and so on. The system tends to a new equilibrium state corresponding to the disintegrated food lattice: $u_{2k-1}=4c$, $u_{2k}=0$, with some constant $c$, but this tendency is not uniform in $n$. For large $n$, almost nothing changes for a long time (it is easy to prove that for the initial data (\ref{VL0}) the derivatives $u^{(j)}_n(0)$ vanish for all $1\le j<n$), so that $u_n$ begins to deviate noticeably from the initial position only after some time. As a result, a binary decay wave is formed, propagating with some velocity $2v$. For negative $t$, the reasoning is similar, but now the food is the neighbor to the left, not to the right. At the moment $t=0$, species $n=1$ is deprived of food and begins to die out; this leads to the extinction of species $n=2$, and so on. As a result, in the decay wave propagating as $t\to-\infty$, all $u_n$ decrease.

\begin{figure}[t!]
\centerline{\includegraphics[width=0.95\textwidth]{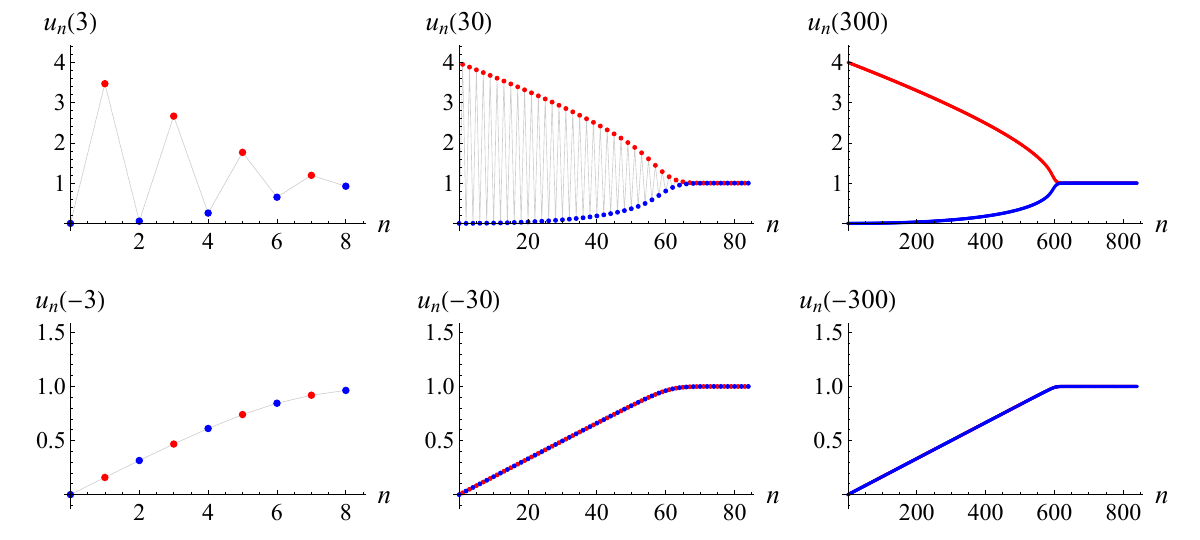}}
\caption{Solution for $\alpha=1$, at several moments $t$. Here and in the following figures, the red dots correspond to odd $n$ and the blue dots to even $n$.}
\label{fig:solution}
\end{figure}

In \cite{Adler_Shabat_2018, Adler_2024}, asymptotic formulas in inverse powers of $t$ were derived for $u_n(t)$ for fixed $n$; they describe well only the region near $n=0$ (the trailing edge of the wave). From these formulas, a new limiting equilibrium position $4c$ is determined (in particular, $c=1$ for the solution in top Fig.~\ref{fig:solution}), but not the velocity $2v$ with which the front edge propagates. The goal of our work is to obtain asymptotic expansions which are uniformly suitable for all $n$. This is possible because the solution becomes asymptotically self-similar for large $t$. Indeed, Fig.~\ref{fig:solution} suggests that the solution profiles at different moments $t$ approximately coincide if we bring them to the same scale along the $n$ axis by passing to the self-similar variable $x=n/(2vt)$. The graph of the solution as a function of $x$ tends to a certain limit curve as $t$ increases. In other words, the solution is asymptotically described by a self-similar formula of the form
\[
 u_n(t)\sim f_0\left(\frac{n}{2vt}\right)~~\text{for odd}~n,\qquad 
 u_n(t)\sim g_0\left(\frac{n}{2vt}\right)~~\text{for even}~n.
\]
Our goal is to determine the velocity $2v$, the functions $f_0$ and $g_0$, and the correction terms.

\begin{remark} 
Throughout this article, the asymptotic expansions as $t\to\pm\infty$ of solutions to the initial-boundary value problem for VL are understood to be exclusively formal asymptotic solutions \cite[Introduction]{Il'in}. For brevity, we will sometimes use the abbreviation FAS.
\end{remark}

Let us briefly outline the content of the paper. In Section \ref{s:series}, we introduce formal asymptotic expansions in inverse powers of $t$, describing the solution as $t\to\pm\infty$ in the self-similar sector. Proposition \ref{th:fg-series}, which is one of the main results of the paper, contains formulas for the leading terms and first-order corrections. The leading terms, which determine the form of the wave, are calculated quite simply; the constants included in them are refined in Section \ref{s:leading} by simple geometric considerations. The derivation of the subsequent coefficients is more difficult and is performed in two steps. Section \ref{s:a1} utilizes the fact that for $\alpha=1$ the solution under study satisfies a certain difference equation consistent with the dynamics in $t$, which allows us to derive a recurrence scheme for calculating the coefficients of the asymptotic series corresponding to this solution. In Section \ref{s:ab}, we apply a substitution between the Volterra and Toda lattices, which makes possible to recalculate the obtained answer for arbitrary $\alpha$. This scheme can be used to calculate any number of FAS coefficients, but in Proposition \ref{th:fg-series} we restrict ourselves to the first two, as the subsequent ones become quite cumbersome.

Numerical experiments (see Section \ref{s:series}) show that the resulting formulas approximate the solution very well, with the exception of a relatively small transition region near the front edge of the wave at $x=1$, where the series coefficients are singular. The asymptotics in this transition region are constructed by passing to a new stretched variable according to the matching method \cite{Il'in}, which leads to series in $t^{-1/3}$. According to the accepted terminology, the asymptotic series in the transition region are called the inner expansion, and the series in the original variables are called the outer expansion. Calculations using the matching method are described in sections \ref{s:x1a1} and \ref{s:x1}, the main result of which are formulas for the inner expansion up to the order $O(t^{-2/3})$. Unlike the outer expansion, whose coefficients are given by elementary functions, the inner expansion requires the use of the special Hastings--McLeod solution of the second Painlev\'e equation \cite{Hastings_McLeod_1980}. This transcendental function arises in the description of the asymptotics despite the fact that for any finite $n$ the solution of the Volterra lattice, as already noted, is expressed in terms of hypergeometric functions.

\begin{remark} 
It is known that transcendental solutions of Painlev\'e equations also appear in the asymptotic formulas for large values of the discrete parameter, even for solutions of discrete Painlev\'e equations which are rational for each individual value of this parameter, see e.g.~\cite{Miller_Sheng_2017}.
\end{remark}

In Section \ref{s:shock}, we discuss shock waves arising in the case of a nonzero boundary condition $u_0(t)=\beta$. Here, various regimes are possible, including those without a self-similar sector, but numerical solutions suggest that if this sector is present, then the leading asymptotic terms in it are the same as in the case of $\beta=0$, only in a narrower interval.

In Section \ref{s:gVL}, we consider a generalization of the problem for lattices of the form $\dot u_n=r(u_n)(u_{n+1}-u_{n-1})$. Apparently, the formation of self-similar decay waves is a fairly universal behavior for them. Recall that, according to Yamilov classification \cite{Yamilov_1984, Yamilov_2006}, lattice equations of this form are integrable (in the sense of existence of higher symmetries, conservation laws, multisoliton solutions, and Lax representations that allow the application of the inverse scattering method) if and only if $r(u)$ is a polynomial of degree no higher than two, which includes the Volterra lattice itself and its modification. In our paper, the fact of integrability of VL is not explicitly used, although some properties that we exploit (the constraint consistent with the dynamics and the substitution into the Toda lattice) are very closely related to integrability.

\section{Asymptotic series}\label{s:series} 

The derivation of the asymptotics is based on the a priori assumption that the decay wave has a self-similar form and propagates with a constant velocity $2v$. We assume that $v>0$ for evolution at $t>0$ and $v<0$ for evolution at $t<0$. Let us introduce the variable
\begin{equation}\label{x}
 x=\frac{n}{2vt},\qquad x\in(0,1)   
\end{equation}
and look for a solution of the Volterra lattice (\ref{VL}) (for now ignoring the boundary and initial conditions) in the form
\begin{equation}\label{uFG}
 u_n(t)= \left\{\begin{aligned} 
  & F(x,t)&& \text{for odd }n,\\
  & G(x,t)&& \text{for even }n,
 \end{aligned}\right.\qquad 0<n<2vt, 
\end{equation}
assuming that $F(x,t)$ and $G(x,t)$ are continuously differentiable with respect to $t$ and analytic with respect to $x$. This leads to the system of functional-differential equations
\begin{equation}\label{FGsys0}
\begin{aligned}
 F_t-\frac{x}{t}F_x &= F\left(G\Bigl(x+\frac{1}{2vt},t\Bigr)
  -G\Bigl(x-\frac{1}{2vt},t\Bigr)\right),\\
 G_t-\frac{x}{t}G_x &= G\left(F\Bigl(x+\frac{1}{2vt},t\Bigr)
  -F\Bigl(x-\frac{1}{2vt},t\Bigr)\right).
\end{aligned}
\end{equation}
Applying the Taylor expansion at $x$, we rewrite this as follows:
\begin{equation}\label{FGsys}
\begin{aligned}
 v(tF_t-xF_x) &= F\left(G_x+\frac{G_{xxx}}{6(2vt)^2}
  +\dotsb+\frac{\partial^{2j+1}_x(G)}{(2j+1)!(2vt)^{2j}}+\dotsb\right),\\
 v(tG_t-xG_x) &= G\left(F_x+\frac{F_{xxx}}{6(2vt)^2}
  +\dotsb+\frac{\partial^{2j+1}_x(F)}{(2j+1)!(2vt)^{2j}}+\dotsb\right).
\end{aligned}
\end{equation}
Further, let us assume that for $t\to\infty$ the asymptotic expansions are applicable
\begin{equation}\label{FGseries}
 F=f_0(x)+\frac{f_1(x)}{t}+\frac{f_2(x)}{t^2}+\dotsb,\qquad 
 G=g_0(x)+\frac{g_1(x)}{t}+\frac{g_2(x)}{t^2}+\dotsb,
\end{equation}
admitting term-by-term differentiation. The expansions for $t\to+\infty$ and $t\to-\infty$ are obviously different (since the solution itself behaves differently), so notations like $f^\pm_j$ and $g^\pm_j$ should be used, but we will not do this to avoid cluttering the formulas; instead, we indicate the sign of $t$ when necessary. Substituting the series into (\ref{FGsys}) and collecting constant terms, we obtain a system of ODEs for $f_0$ and $g_0$:
\begin{equation}\label{fg0sys}
 vxf'_0+f_0g'_0=0,\qquad g_0f'_0+vxg'_0=0.
\end{equation}
By collecting the coefficients of $t^{-1}$, we obtain
\begin{equation}\label{fg1sys}
 vxf'_1+f_0g'_1+(g'_0+v)f_1=0,\qquad g_0f'_1+vxg'_1+(f'_0+v)g_1=0;
\end{equation}
the coefficients of $t^{-2}$ give
\begin{equation}\label{fg2sys}
\begin{aligned}
 & vxf'_2+f_0g'_2+(g'_0+2v)f_1+f_1g'_1+\frac{1}{24v^2}f_0g'''_0=0,\\ 
 & g_0f'_2+vxg'_2+(f'_0+2v)g_2+f'_1g_1+\frac{1}{24v^2}f'''_0g_0=0,    
\end{aligned}
\end{equation}
and so on: it is easy to show that in the order $t^{-j}$ a system arises
\[
 \begin{pmatrix}
  vx & f_0\\
  g_0 & vx
 \end{pmatrix}\begin{pmatrix} f'_j \\ g'_j\end{pmatrix}
 +\begin{pmatrix}
  g'_0+jv & 0\\
  0 & f'_0+jv   
\end{pmatrix}\begin{pmatrix} f_j \\ g_j\end{pmatrix}=\dotsc,
\]
where the right-hand side is a vector with components which are differential polynomials in $f_i$ and $g_i$ for $i<j$. The matrix at the derivatives on the left-hand side is the same for all $j$. It must be singular, since otherwise the system (\ref{fg0sys}) for the leading coefficients has only a constant solution, which is of no interest. Despite this, the leading coefficients themselves are determined almost uniquely.

\begin{proposition}\label{th:fg}
Any non-constant solution of the system \textnormal{(\ref{fg0sys})} has the form
\begin{equation}\label{fgc}
 f_0(x)=2c-vx+2\sigma\sqrt{c(c-vx)},\qquad  g_0(x)=2c-vx-2\sigma\sqrt{c(c-vx)},
\end{equation}
with some constant $c$ and $\sigma=\pm1$.
\end{proposition}

\begin{proof}
A non-constant solution is only possible if 
\begin{equation}\label{fg1}
 f_0g_0=v^2x^2.
\end{equation}
On the other hand, adding up the equations (\ref{fg0sys}), we obtain $vx(f'_0+g'_0)=-f_0g'_0-f'_0g_0=-2v^2x$. It follows that $f'_0+g'_0=-2v$, that is,
\begin{equation}\label{fg2}
 f_0+g_0=4c-2vx.
\end{equation}
By solving the obtained algebraic equations, we arrive at (\ref{fgc}), and it is directly proved that these formulas satisfy (\ref{fg0sys}).
\end{proof}

As for the subsequent coefficients $f_j$ and $g_j$, it is impossible to find them unambiguously from the equations written out. In particular, it is easy to prove, taking into account formulas (\ref{fg1}) and (\ref{fg2}), that equations (\ref{fg1sys}) are equivalent to a single algebraic relation
\[
 vxf_1+f_0g_1=0,
\]
that is, one of the functions $f_1$ or $g_1$ is arbitrary. The equations (\ref{fg2sys}) are reduced to equations of the form $f_2+g_2=\dots$ and $f'_2=\dots$ with right-hand sides expressed in terms of the previous coefficients; therefore, at this step, the coefficients are found up to one integration constant. The 
same is true for further coefficients: the equations for $f_j$ and $g_j$ for $j>2$ also reduce to an algebraic equation and a differential one.

Thus, the coefficients of the series (\ref{FGseries}) contain, generally speaking, one arbitrary function and an infinite set of integration constants. This is to be expected, since so far we have not taken into account the initial data for the solution that should be represented by these series. The leading terms (\ref{fgc}) turn out to be universal, but it is quite natural that the subsequent corrections depend on the solution under consideration. It turns out that for the special initial-boundary data of interest to us, the coefficients $f_1$ and $g_1$ (and, in principle, all subsequent ones) can be calculated exactly.

\begin{proposition}\label{th:fg-series}
For the initial boundary value problem \textnormal{(\ref{VL+}), (\ref{VL0}):}

1) if $t>0$ then
\[
 v=1,\qquad 4c=\alpha+\alpha^{-1}+2,\qquad x=n/(2t)
\] 
and the first two coefficients of FAS \textnormal{(\ref{FGseries})} are
\begin{gather}
\label{fg0p}
 f_0(x)=2c-x+2\sqrt{c(c-x)},\qquad 
 g_0(x)=2c-x-2\sqrt{c(c-x)},\\[5pt]
\label{fg1p}
\begin{aligned}
 f_1(x)&=\frac{\bigl(\sqrt{1-x}-2+2x\bigr)\bigl(\sqrt{c(c-x)}+c-x\bigr)-x(1-x)}
  {4(c-x)(1-x)},\\
 g_1(x)&=-\frac{\sqrt{1-x}\bigl(\sqrt{c(c-x)}-c+x\bigr)-x(1-x)}{4(c-x)(1-x)};
\end{aligned}
\end{gather} 

2) if $t<0$ then 
\[
 v=-1,\qquad 4c=\alpha+\alpha^{-1}-2,\qquad x=n/(-2t)
\] 
and the first two coefficients of FAS \textnormal{(\ref{FGseries})} are
\begin{gather}
\label{fg0m}
 f_0(x)=2c+x+2\sigma\sqrt{c(c+x)},\qquad 
 g_0(x)=2c+x-2\sigma\sqrt{c(c+x)},\\[5pt]
\label{fg1m}
\begin{aligned}
 f_1(x)&=\frac{\bigl(1-2\sqrt{1-x}\bigr)\bigl(x+2c+2\sigma\sqrt{c(c+x)}\bigr)+x}
  {8(c+x)\sqrt{1-x}},\\
 g_1(x)&=\frac{c+x-x\sqrt{1-x}-\sigma\sqrt{c(c+x)}}{4(c+x)\sqrt{1-x}}
\end{aligned}
\end{gather} 
where $\sigma=1$ if $\alpha\ge1$ and $\sigma=-1$ if $\alpha<1$.    
\end{proposition}

From here on, a square root of a positive number is always understood to be positive. Note that all roots involved in the above formulas are correctly defined on the interval $x\in[0,1]$, since the inequality $c\ge 1$ holds for $t>0$, and the inequality $c\ge0$ holds for $t<0$.

The proof of the validity of Proposition \ref{th:fg-series} is carried out in several stages. In section \ref{s:leading}, we analyze the formulas (\ref{fgc}) and determine the constants $v$, $c$, and the sign of $\sigma$ included in them, which completely determines the leading coefficients $f_0$ and $g_0$. In section \ref{s:a1}, corrections are calculated for the case $\alpha=1$. This is possible because the solution with the initial data under consideration satisfies a certain constraint which is consistent with the dynamics with respect to $t$. This constraint gives an explicit recurrence scheme for computing $f_j$ and $g_j$, which turn out to be rational functions of $\sqrt{1-x}$. Several coefficients are given in Proposition \ref{th:alpha1-series}. Section \ref{s:ab} gives a calculation for the case $\alpha\ne1$, which is possible due to the B\"acklund transformation connecting it with the case $\alpha=1$. This also gives explicit recursive formulas for $f_j$ and $g_j$. They turn out to be more cumbersome rational functions of the two roots $\sqrt{1-x}$ and $\sqrt{c(c-vx)}$, so in Proposition \ref{th:fg-series} we restrict ourselves to the first corrections. As a side result, Proposition \ref{th:ab-series} provides an asymptotic expansion for the Toda lattice.

\begin{figure}[t!]
\centerline{\includegraphics[width=0.9\textwidth]{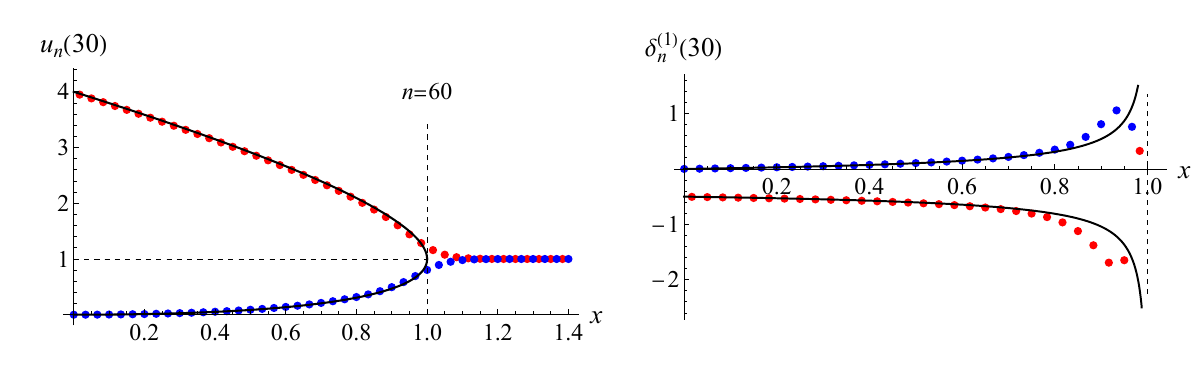}}
\centerline{\includegraphics[width=0.9\textwidth]{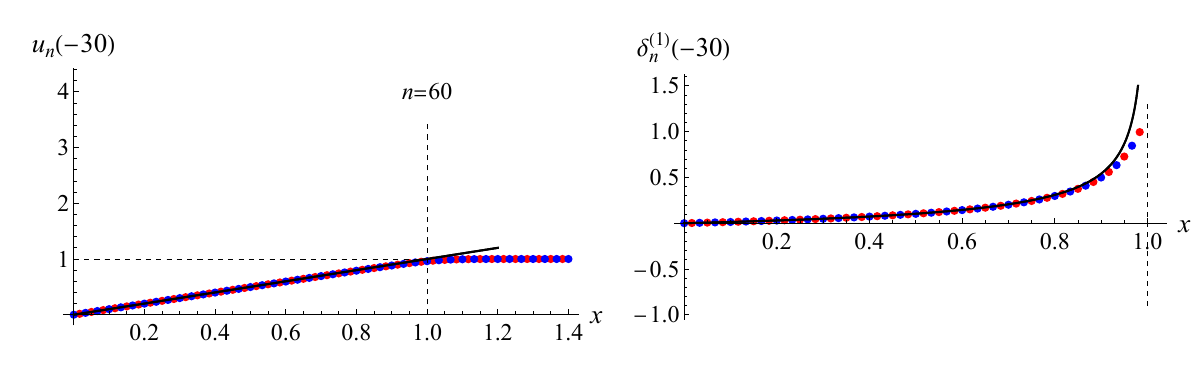}}
\caption{Numerical solution and the first correction term for $\alpha=1$.}
\label{fig:gr1} 
\end{figure}

\begin{figure}[t!]
\centerline{\includegraphics[width=0.9\textwidth]{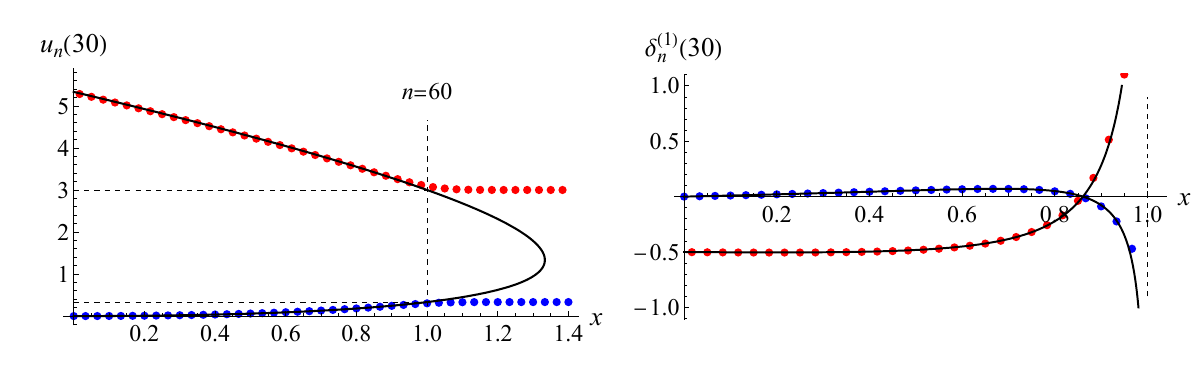}}
\centerline{\includegraphics[width=0.9\textwidth]{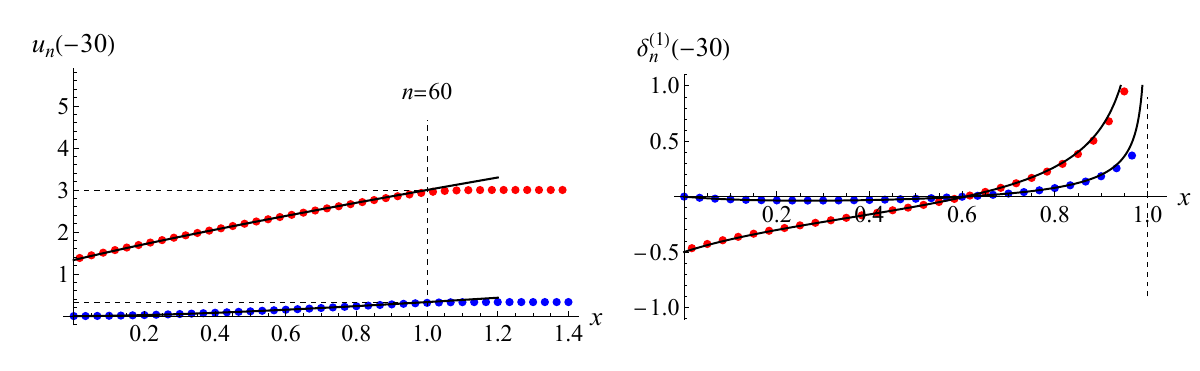}}
\caption{Numerical solution and the first correction term for $\alpha=3$.}
\label{fig:gr3}
\end{figure}

\begin{figure}[t!]
\centerline{\includegraphics[width=0.9\textwidth]{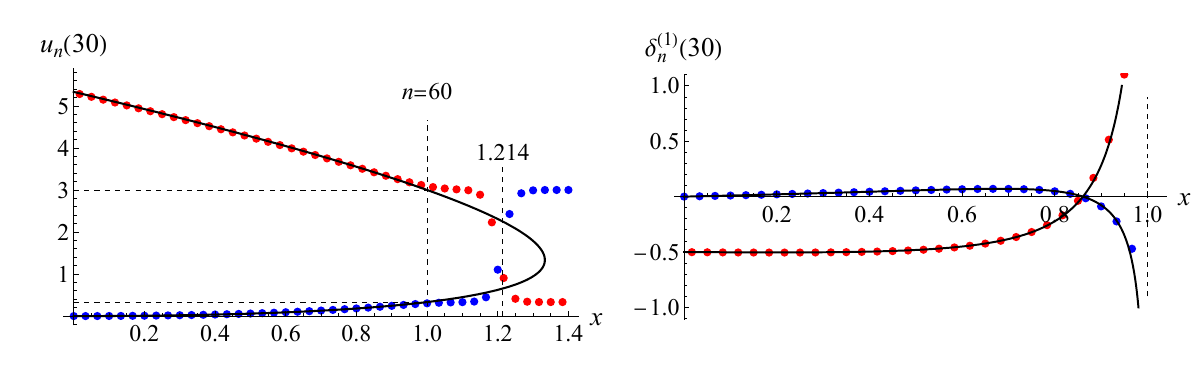}}
\centerline{\includegraphics[width=0.9\textwidth]{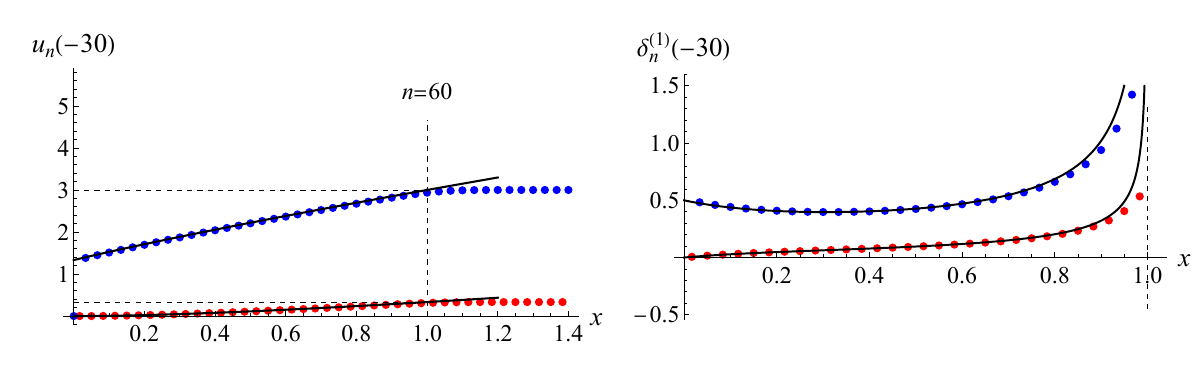}}
\caption{Numerical solution and the first correction term for $\alpha=1/3$.}
\label{fig:gr13} 
\end{figure}

Figs.~\ref{fig:gr1}, \ref{fig:gr3} and \ref{fig:gr13} show the results of numerical calculations for the values $\alpha=1$, $3$ and $1/3$, respectively, illustrating all the qualitatively different solution types. Graphs on the left present the solutions themselves for $t=\pm30$, referred to the variable $x$. The black curves are the plots of the leading terms $f_0(x)$ and $g_0(x)$, which, as we see, closely approximate the solution in the range $0\le x<1$ even for such comparatively small values of $t$. Graphs on the right present, for the same moments of $t$, the quantities
\[
 \delta^{(1)}_n(t)= \left\{\begin{array}{ll}
   t(u_n(t)-f_0(x)) & \text{for odd}~n,\\
   t(u_n(t)-g_0(x)) & \text{for even}~n,\end{array}\right.
\]
combined with the plots of the functions $f_1(x)$ and $g_1(x)$. Here, too, excellent agreement is observed. We emphasize that it makes sense to consider these functions only for values $x<1$, for which the expansions (\ref{FGseries}) are constructed. At the point $x=1$, the corrections themselves have a singularity, which means that a detailed description of solutions near this point requires other expansions. This analysis is carried out further in section \ref{s:x1a1} for the case $\alpha=1$ and in section \ref{s:x1} for arbitrary $\alpha$.

\section{Analysis of the leading terms}\label{s:leading} 

The constants $c$, $v$, and the sign $\sigma$ in (\ref{fgc}) are refined by considering the limiting values for $f_0(x)$ and $g_0(x)$ as $x\to0$ and $x\to1$, under the assumption that the self-similar ansatz (\ref{uFG}) is applicable on the entire interval $x\in(0,1)$ and that near the point $n=2vt$, that is, at $x=1$, the solution reaches the sequence of initial data consisting of alternating $\alpha$ and $\alpha^{-1}$. Substituting these values into the equalities (\ref{fg1}) and (\ref{fg2}) gives
\[
 v^2=f_0(1)g_0(1)=1,\quad 4c=f_0(1)+g_0(1)+2v=\alpha+\alpha^{-1}+2v,
\]
which leads to the formulas for $v$ and $c$ given in Proposition \ref{th:fg-series}.

The graphs of $f_0(x)$ and $g_0(x)$ are two branches of a parabola lying sideways and tangent to the $x$-axis at the origin. If $t>0$ then $v>0$ and the branches of the parabola point to the left, while the rightmost point of the parabola has the abscissa $x=c$; if $t<0$ then $v<0$ and the branches point to the right. In the case $\alpha=1$ and $t<0$, the parabola degenerates into a straight line. We are interested only in the part of these graphs between the lines $x=0$ and $x=1$.

If $\sigma=1$, then $f_0$ corresponds to the upper branch of the parabola and $g_0$ to the lower one; if $\sigma=-1$, then vice versa. To determine the sign of $\sigma$, consider the values of the functions at $x=0$:
\[
 f_0(0)=2(c+\sigma\sqrt{c^2}),\qquad g_0(0)=2(c-\sigma\sqrt{c^2}).
\]
One number is $0$ and the other is $4c\ge0$. Here we have to consider several cases depending on the sign of $t$ and the value of the parameter $\alpha$. 

\begin{figure}[t]
\centerline{
\includegraphics[height=0.2\textheight]{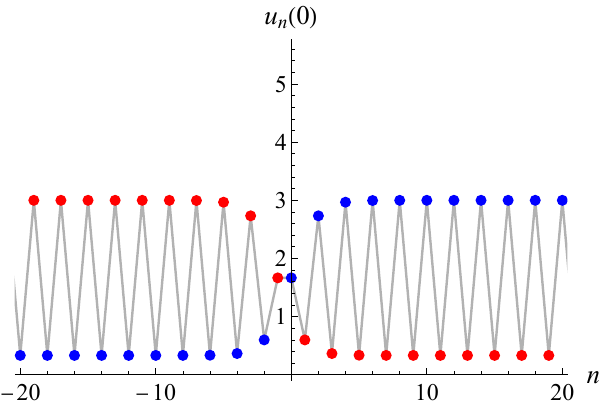}\qquad
\includegraphics[height=0.2\textheight]{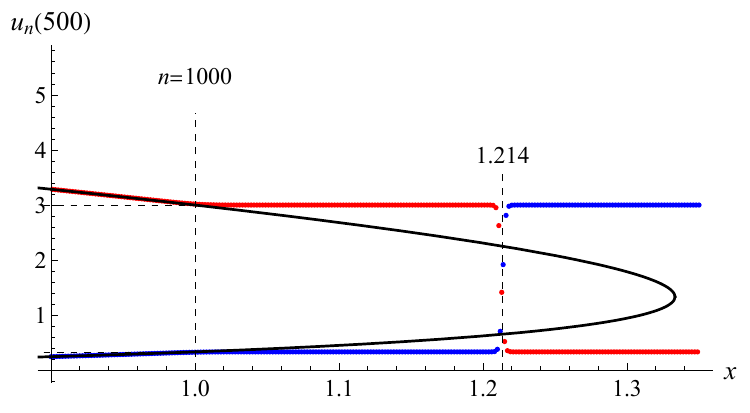}}
\caption{Left: bipartite kink (\ref{VL_kink}) for $t=0$, $\alpha=1/3$ and $\delta=0$.
Right: limiting position of the inversion point for $\alpha=1/3$.}
\label{fig:VL_kink} 
\end{figure}

First, let $t>0$. Since $\dot u_1=u_1u_2>0$, the variable $u_1(t)$ increases monotonically and its limit value $f_0(0)$ cannot be equal to 0. Therefore, $f_0(0)=4c$, $g_0(0)=0$, and $\sigma=1$, which yields expressions (\ref{fg0p}) for $f_0(x)$ and $g_0(x)$. These formulas are invariant with respect to the change $\alpha\to 1/\alpha$ (as well as the expressions for the subsequent coefficients, which will be clear from what follows). According to these formulas, $f_0(1)\ge g_0(1)$, that is $f_0(1)=\mu=\max(\alpha,1/\alpha)$ and $g_0(1)=1/\mu$. If $\alpha\ge1$, then this agrees with the initial data $u_n(0)=\alpha$ for odd $n$ and $u_n(0)=1/\alpha$ for even $n$, and near $x=1$ the solution departs directly to the steady state, as shown in the left top Fig.~\ref{fig:gr3}. If $\alpha<1$, then near $x=1$ the solution approaches ``alien'' constants, and in this case it must switch somewhere to the initial steady state. The numerical solution in the left top Fig.~\ref{fig:gr13} shows that this indeed happens at some point $x_*>1$. Thus, the solutions corresponding to the parameters $\alpha$ and $1/\alpha$ have the same asymptotic expansions at $t\to+\infty$ in the self-similar sector (this explains why the graphs of the first corrections presented in the right top Fig.~\ref{fig:gr3} and Fig.~\ref{fig:gr13} are visually indistinguishable), but differ topologically outside it.

In more detail, the inversion of equilibrium positions for even and odd lattice sites can be described using the fact that the lattice equation (\ref{VL}) admits an exact elementary solution in the form of a bipartite kink (see Fig.~\ref{fig:VL_kink}):
\begin{equation}\label{VL_kink}
 u_{2k}=\frac{1}{u_{2k+1}} = 
 \frac{1+\alpha^{2k+2}e^{\gamma t+\delta}}
      {\alpha(1+\alpha^{2k}e^{\gamma t+\delta})},\qquad 
 \gamma=\alpha^{-1}-\alpha.
\end{equation}
This formula can also be written as
\[
 u_{2k}=\frac{1}{u_{2k+1}}
  =\frac{1}{2}(\alpha^{-1}+\alpha)
   -\frac{\gamma}{2}\tanh\frac{\gamma t+2k\log\alpha+\delta}{2}.
\]
Assuming that $\alpha<1$ we have $u_{2k}=\frac{1}{u_{2k+1}}\to\alpha^{\mp1}$ at $k\to\pm\infty$. The velocity of the kink is
\[
 V_{\text{kink}}=-\frac{\gamma}{\log\alpha}=\frac{\alpha-\alpha^{-1}}{\log\alpha}>0,
\]
that is, it moves to the right, in contrast to the solitons of VL which move to the left. It is natural to assume that the matching of the values $f_0(1)=\alpha^{-1}$ and $g_0(1)=\alpha$ with the limiting values $u_{2k-1}=\alpha$ and $u_{2k}=\alpha^{-1}$ as $k\to\infty$ occurs precisely with the help of the kink (\ref{VL_kink}) which gives exponentially small corrections in the self-similar sector. We can say that in this case a certain hybrid of the self-similar solution (\ref{uFG}) and the traveling-wave solution arises. The velocities of both solutions are constant, but the first solution propagates due to the scaling transformation, and the second one due to the translation. It follows that asymptotically the abscissa of the inversion point in the relative variable $x=n/(2t)$ is determined as the ratio of these velocities
\begin{equation}\label{x*}
 x_*=\frac{V_{\text{kink}}}{2}=\frac{\alpha-\alpha^{-1}}{2\log\alpha}.
\end{equation}
For $\alpha=1/3$ we have $x_*\approx1.214$, which is in good agreement with the numerical solution already at $t=30$ in Fig.~\ref{fig:gr13}. The same solution is shown on the right in Fig.~\ref{fig:VL_kink} for $t=500$, on the interval $x\in[0.9,1.35]$. The absolute width of the flip region does not depend on $t$, therefore its relative width tends to 0 with increasing $t$. Note that for all $\alpha\in(0,1)$ the inequalities are true
\[
 1\le x_* \le c=\frac{1}{4}(\alpha+\alpha^{-1}+2),
\]
meaning that the inversion point lies inside the parabola.

For $t<0$, the solutions corresponding to $\alpha$ and to $1/\alpha$ differs in a simpler way, just by switching variables with even and odd $n$ between the two branches of the parabola (cf.~bottom Figs.~\ref{fig:gr3} and \ref{fig:gr13}). As $-t$ increases, both variables $u_1(t)$ and $u_2(t)$ decrease without changing their relative positions. If $\alpha\ge1$, then $u_1(t)\ge u_2(t)$, and then $f_0(0)=4c$ and $g_0(0)=0$, that is, $\sigma=1$. If $\alpha<1$, then, conversely, $u_1(t)<u_2(t)$ and then $f_0(0)=0$ and $g_0(0)=4c$, which means $\sigma=-1$. As a result, we arrive at the formulas (\ref{fg0m}).

\section{Calculation of correction terms in the case \texorpdfstring{$\alpha=1$}{alpha=1}}
\label{s:a1} 

Let us return to the series (\ref{FGseries}). As we have seen, calculating subsequent coefficients requires additional information about the solution. In the simplest case $\alpha=1$, this can be obtained using the following property.

\begin{proposition}
The Volterra lattice \textnormal{(\ref{VL})} is consistent with the constraint
\begin{equation}\label{constr}
 (tu_{n-1}+tu_n+n)(tu_n+tu_{n+1}+n+1)u_n=(2tu_n+n)(2tu_n+n+1)
\end{equation}
and the solution corresponding to the boundary condition $u_0(t)=0$ and the initial data $u_n(0)=1$ for $n>0$ satisfies this constraint for all $t$.
\end{proposition}

To prove it, it suffices to differentiate (\ref{constr}) by virtue of (\ref{VL}) and verify by direct calculation that the resulting relation is an algebraic consequence of (\ref{constr}). It follows that if some solution of VL satisfies the constraint at some time $t$, then this is true for all $t$ for which this solution is defined. At time $t=0$, the constraint equation simplifies to
\[
 n(n+1)(u_n(0)-1)=0,
\]
which is satisfied by the initial data under study, therefore the solution with such initial data satisfies (\ref{constr}) for any $t$; the boundary condition $u_0(t)=0$ is also preserved automatically due to the dynamics.

\begin{remark}
The constraint (\ref{constr}) admits a generalization with arbitrary parameters $c_0,c_1,c_2$ and $c_3$ obtained in \cite{Adler_Shabat_2019} as the stationary equation for certain non-autonomous symmetry of VL:
\begin{equation}\label{cc}
\begin{aligned}
 &(tu_{n-1}+tu_n+n+c_0)(tu_n+tu_{n+1}+n+c_0+1)u_n\\
 &\qquad\qquad = c_1(4tu_n+2n+2c_0+1)^2 +(-1)^nc_2(4tu_n+2n+2c_0+1) +c_3.
\end{aligned}
\end{equation}
This makes possible to handle slightly more general initial data (not including, however, (\ref{VL0}) for $\alpha\ne1$). A related, simpler constraint was studied in \cite{Its_Kitaev_Fokas_1990, Fokas_Its_Kitaev_1991}. Equation (\ref{cc}) coincides, up to notation, with the discrete Painlev\'e equation dP-XXXIV, and reduces the Volterra lattice to the Painlev\'e equation P-V, up to a simple substitution. Moreover, if the parameters are consistent with the zero boundary condition $u_0(t)=0$ (like, in particular, for the relation (\ref{constr})), then the equation for $u_1(t)$ becomes a Riccati equation and its solution is expressed by the logarithmic derivative of the confluent hypergeometric function. Using this, it is possible to obtain asymptotic formulas for fixed $n$ \cite{Adler_Shabat_2018, Adler_Shabat_2019}. However, to obtain self-similar asymptotics, it is more convenient to work directly with the constraint equation.
\end{remark}

\begin{proposition}\label{th:alpha1-series}
For $\alpha=1$, the first few coefficients of the asymptotic series \textnormal{(\ref{FGseries})} for $t\to+\infty$ and for odd $n$ are:
\begin{gather*}
 f_0(x)=2-x+2\sqrt{1-x},\qquad 
 f_1(x)=-\frac{1}{4}\left(1+\frac{1}{\sqrt{1-x}}\right),\\
 f_2(x)=-\frac{1}{32(1-x)^2}\left(1+\frac{1}{\sqrt{1-x}}\right),\qquad
 f_3(x)=-\frac{1}{64(1-x)^3}\left(1+\frac{4+x}{4\sqrt{1-x}}\right),\\
 f_4(x)=-\frac{1}{2^{11}(1-x)^5}\left(25+11x+\frac{50+24x-x^2}{2\sqrt{1-x}}\right),\\
 f_5(x)=-\frac{1}{2^{11}(1-x)^6}\left(26+19x+\frac{416+366x+21x^2}{16\sqrt{1-x}}\right),\\
 f_6(x)=-\frac{1}{2^{16}(1-x)^8}\left(1073+1400x+173x^2
   +\frac{4292+5704x+681x^2-20x^3}{4\sqrt{1-x}}\right),\dotsc,
\end{gather*} 
and the coefficients $g_j(x)$ corresponding to even $n$ are obtained from $f_j(x)$ by changing the sign of the root $\sqrt{1-x}$; 

for $t\to-\infty$ we have
$g_j(x)=f_j(x)$ and
\begin{gather*}
 f_0(x)=x,\qquad 
 f_1(x)=-\frac{1}{4}\left(1-\frac{1}{\sqrt{1-x}}\right),\\
 f_2(x)=\frac{1}{32(1-x)^2}\left(1-\sqrt{1-x}\right),\qquad
 f_3(x)=-\frac{1}{64(1-x)^3}\left(1-\frac{4+x}{4\sqrt{1-x}}\right),\\
 f_4(x)= \frac{1}{2048(1-x)^5}\left(25+11x-5(5+2x)\sqrt{1-x}\right),\\
 f_5(x)= -\frac{1}{2^{11}(1-x)^6}\left(26+19x-\frac{416+366x+21x^2}{16\sqrt{1-x}}\right),\\ f_6(x)= \frac{1}{2^{18}(1-x)^8}\left(4(1073+1400x+173x^2)-(4292+5496x+651x^2)\sqrt{1-x}\right),\dotso.
\end{gather*} 
\end{proposition}

\begin{proof}
We replace the variables $u_n$ in (\ref{constr}) according to the formulas (\ref{uFG}). This gives a pair of equations: for even $n$ it reads
\begin{gather*}
 \left(tF\Bigl(x-\frac{1}{2vt},t\Bigr)+tG(x,t)+2vxt\right)
 \left(tG(x,t)+tF\Bigl(x+\frac{1}{2vt},t\Bigr)+2vxt+1\right)G(x,t)\\
 -(2tG(x,t)+2vxt)(2tG(x,t)+2vxt+1)=0
\end{gather*}
and equation for odd numbers is obtained by interchanging $F$ and $G$. Next, we substitute the expansions (\ref{FGseries}) and collect terms with the same powers of $t$. The leading coefficients of the expansions are already known. For $t\to+\infty$, we have
\[
 v=1,\qquad f_0=2-x+2\sqrt{1-x},\qquad g_0=2-x+2\sqrt{1-x}, 
\]
according to the formula (\ref{fg0p}) for $\alpha=1$. The terms with $t^2$ are canceled out identically, and in the subsequent orders systems of linear algebraic equations with respect to $f_j$ and $g_j$ arise, of the form
\[
 (2-x)f_j+(2-x+2\sqrt{1-x})g_j=\dotsc,\qquad
 (2-x-2\sqrt{1-x})f_j+(2-x)g_j=\dotsc
\]   
where right-hand sides are expressions containing $f_i$ and $g_i$ (and their derivatives) for $i<j$. From here, the coefficients can be found recursively.

For $t\to-\infty$ we have
\[
 v=-1,\qquad f_0=g_0=x,
\]
according to the formula (\ref{fg0m}) for $\alpha=1$. In this case, both the terms with $t^2$ and the terms with $t^1$ cancel identically. The terms with $t^0$ give a system of algebraic equations for $f_1$ and $g_1$, the solution of which is written above, with the choice of the square root sign determined by the fact that $f_1$ and $g_1$ must be positive (since $F=x+f_1(x)/t+\dots$ approaches $x$ from below and $t<0$). In subsequent orders, we obtain equations of the form
\[
 (x-2)f_j+xg_j=\dotsc,\qquad xf_j+(x-2)g_j=\dotsc,
\]   
with $f_i$ and $g_i$ for $i<j$ in the right-hand side, from which, again, coefficients are found recursively. By construction, the equations are invariant under the permutation of $f$ and $g$, and since in this case $f_0=g_0$ and $f_1=g_1$, then the equalities $f_j=g_j$ are satisfied for all $j$. This can be taken into account in the formulas to simplify the calculations.
\end{proof}

\begin{figure}[t]
\centerline{\includegraphics[width=0.95\textwidth]{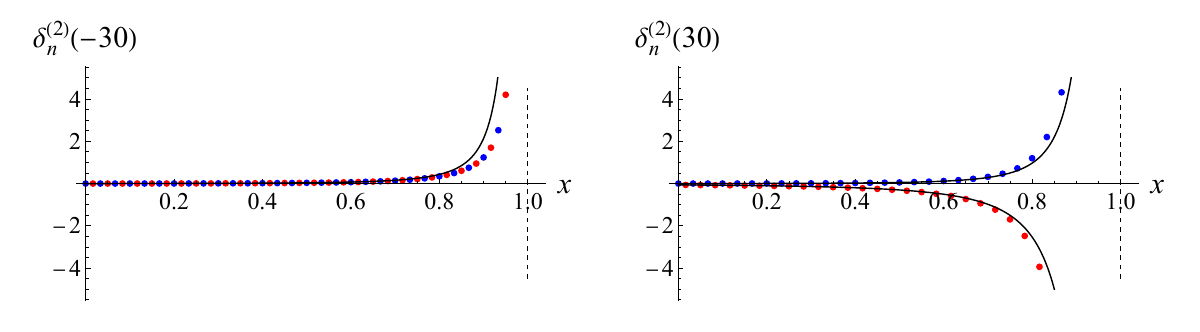}} 
\caption{Second correction for $\alpha=1$, at $t=\pm30$.}
\label{fig:delta2} 
\end{figure}

The graphs of the coefficients $f_2(x)$, $g_2(x)$ are shown in Fig.~\ref{fig:delta2}, along with the points
\[
\delta^{(2)}_n(t)=\left\{\begin{array}{ll}
  t^2(u_n(t)-f_0(x))-tf_1(x) & \text{for odd}~n,\\
  t^2(u_n(t)-g_0(x))-tg_1(x) & \text{for even}~n
 \end{array}\right.
\]
computed for the numeric solution $u_n(t)$.

\section{Asymptotics near the point \texorpdfstring{$x=1$}{x=1} for \texorpdfstring{$\alpha=1$}{alpha=1}}\label{s:x1a1}  

Proposition \ref{th:alpha1-series} implies that the asymptotic expansions (\ref{FGseries}) lose their applicability as $x$ approaches $x=1$, since the coefficients $f_j(x)$ have singularities at this point, and their order increases with increasing index $j$. According to the method of matching asymptotic expansions \cite{Il'in}, to correctly describe the solution in such a situation, one should pass from $x$ to a stretched variable
\begin{equation}\label{zr}
 z=(x-1)|t|^{2r},\qquad r>0.
\end{equation}
The value of $r$ is chosen based on the following considerations. Let us perform such a transformation in the formulas from Proposition \ref{th:alpha1-series} and retain the leading terms in $t$. First, we consider the series as $t\to+\infty$. For them, we obtain, for finite $z<0$:
\begin{alignat*}{2}
 f_0 &= 1+2t^{-r}(-z)^{1/2}(1+O(t^{-r})),&
 \frac{f_1}{t} &= -t^{r-1}\frac{(1+O(t^{-r}))}{4(-z)^{1/2}},\\
 \frac{f_2}{t^2} &= -t^{-r+2(3r-1)}\frac{(1+O(t^{-r}))}{32(-z)^{5/2}},&
 \frac{f_3}{t^3} &= -t^{r-1+2(3r-1)}\frac{5(1+O(t^{-r}))}{2^8(-z)^{7/2}},\\
 \frac{f_4}{t^4} &= -t^{-r+4(3r-1)}\frac{73(1+O(t^{-r}))}{2^{12}(-z)^{11/2}},&
 \frac{f_5}{t^5} &= -t^{r-1+4(3r-1)}\frac{803(1+O(t^{-r}))}{2^{15}(-z)^{13/2}},\\
 \frac{f_6}{t^6} &= -t^{-r+6(3r-1)}\frac{10657(1+O(t^{-r}))}{2^{18}(-z)^{17/2}},\qquad &&\dotso.
\end{alignat*} 
Hence, representing the first of the expansions (\ref{FGseries}) in the form
\[
 F(x,t)=\sum_{j=0}^{\infty}\frac{f_{2j}}{t^{2j}}+\sum_{j=0}^{\infty}\frac{f_{2j+1}}{t^{2j+1}},
\]
we obtain, after passing to the variable (\ref{zr}),
\begin{gather*}
 F=1+t^{-r}\left[2\sqrt{-z}-\frac{t^{2(3r-1)}}{32(-z)^{5/2}}-\frac{73t^{4(3r-1)}}{2^{12}(-z)^{11/2}}
   -\frac{10657t^{6(3r-1)}}{2^{18}(-z)^{17/2}}+\dotsb\right](1+O(t^{-r}))\\
   -t^{r-1}\left[\frac{1}{4\sqrt{-z}}+\frac{5t^{2(3r-1)}}{128(-z)^{7/2}}
   +\frac{803t^{4(3r-1)}}{2^{15}(-z)^{13/2}}+\dotsb\right](1+O(t^{-r})).
\end{gather*}
The constant $r$ must be chosen in such a way (see, for example, the arguments given in \S IV.3.2 of the monograph \cite{Il'in}) that in this formula all the terms written out in both square brackets have the order $O(1)$ for $t\to\infty$ and finite negative values of $z$. Thus, we should set $r=1/3$ and apply the substitution
\begin{equation}\label{z}
 z=(x-1)t^{2/3}.
\end{equation}
Then all the terms of the series $F=f_0+f_1/t+f_2/t^2+\dotsb$ with even numbers will be of the order of $t^{-r}=t^{-1/3}$, and terms with odd numbers will be  of the order of $t^{r-1}=t^{-2/3}$. After formal rearranging the terms, we come to a series of the form
\begin{equation}\label{aspf} 
 F=1+t^{-1/3}p_1(z)+t^{-2/3}p_2(z)+t^{-1}p_3(z)+\dotsb
\end{equation} 
with the coefficients which, in turn, are represented by formal series in inverse powers of $\sqrt{-z}$. Similarly, the formal asymptotic series for $G$ is transformed into a series
\begin{equation}\label{aspg} 
 G=1+t^{-1/3}q_1(z)+t^{-2/3}q_2(z)+t^{-1}q_3(z)+\dotsb,
\end{equation} 
where, according to Proposition \ref{th:alpha1-series}, $q_j$ differ from $p_j$ by replacing $\sqrt{-z}$ with $-\sqrt{-z}$. The first coefficients of the series $F$ and $G$ are determined from the leading terms of the expansions for $f_{2j}/t^{2j}$ and are represented by the series
\begin{equation}\label{pminf}
 p_1(z)=\sqrt{-z}\left(2-\frac{1}{32(-z)^3}-\frac{73}{2^{12}(-z)^6}-\frac{10657}{2^{18}(-z)^9}+\dotsb\right),\qquad q_1(z)=-p_1(z).
\end{equation}

Thus, according to the matching method, asymptotic expansions suitable for all finite values of $z$ as $t\to+\infty$ should be sought in the form of series (\ref{aspf}) and (\ref{aspg}) satisfying the following requirements:

1) as $z\to-\infty$, the coefficients $p_1(z)$ and $q_1(z)$ should be expanded into power asymptotic series (\ref{pminf}), which ensures consistency with the expansions (\ref{FGseries}) in the region $x<1$;

2) as $z\to\infty$, all coefficients $p_j(z)$ and $q_j(z)$ must be rapidly decreasing, that is, they must tend to zero along with all derivatives, which ensures consistency of the expansions with the solution of VL tending to the steady state $u_n(t)=1$ in the region $x>1$.

When passing from $x$ to the stretched variable (\ref{z}), the system of equations (\ref{FGsys0}) takes the form
\begin{equation}\label{FGz-sys}
\begin{aligned}
 tF_t-t^{2/3}F_z-\frac{z}{3}F_z 
  &=\frac{t^{2/3}}{v}F\left(G_z+\frac{G_{zzz}}{24v^2t^{2/3}}+\dotsb\right),\\
 tG_t-t^{2/3}G_z-\frac{z}{3}G_z 
  &=\frac{t^{2/3}}{v}G\left(F_z+\frac{F_{zzz}}{24v^2t^{2/3}}+\dotsb\right). 
\end{aligned}
\end{equation}
Here we set $v=1$, substitute the series (\ref{aspf}) and (\ref{aspg}), and collect the terms with powers of $t^{1/3}$, $t^0$ and $t^{-1/3}$, which gives the following set of equations for the coefficients:
\begin{gather}
\label{pq-sys1}
 -p'_1=q'_1,\qquad -q'_1=p'_1, \\
\label{pq-sys2}
 -p'_2=q'_2+p_1q'_1,\qquad 
 -q'_2=p'_2+q_1p'_1,\\
\label{pq-sys3}
 \begin{gathered}
 -p'_3-\frac{(zp_1)'}{3}= q'_3+\frac{q'''_1}{24}+p_1q'_2+p_2q'_1,\qquad
 -q'_3-\frac{(zq_1)'}{3}= p'_3+\frac{p'''_1}{24}+q_1p'_2+q_2p'_1.
 \end{gathered}
\end{gather}
It is easy to see that equations (\ref{pq-sys1}) and (\ref{pq-sys2}) can be integrated and, taking into account the condition of rapid decrease of their solutions as $z\to\infty$, they are reduced to
\begin{equation}\label{q12}
 q_1=-p_1,\qquad q_2=-p_2+\frac{p_1^2}{2}.
\end{equation}
We substitute $q_1$ and $q_2$ according to this and replace equations (\ref{pq-sys3}) with their sum and difference, which gives
\[
 p'_3+q'_3= p_1p'_2+p'_1p_2-\frac{3}{4}p^2_1p'_1,\qquad p'''_1-8(zp_1)'-6p^2_1p_1'=0.
\]
After integration, taking into account the condition of rapid decrease of $p_j$ and $q_j$, we obtain
\begin{equation}\label{q3}
 q_3=-p_3+p_1p_2-\frac{p_1^3}{4},\qquad p''_1=2p^3_1+8zp_1.
\end{equation}
A closed ODE has arisen for the function $p_1$, and given the boundary conditions it satisfies, we conclude that $p_1$ is the special Hastings--McLeod solution of the second Painlev\'e equation (P-II) with zero parameter
\begin{equation}\label{P-II}
 y''=2y^3+8zy,
\end{equation}
satisfying the conditions
\begin{equation}\label{HM-asymp}
 y(z)=2(-z)^{1/2}(1+o(1)),~~ z\to-\infty,\qquad y(z)\to0,~~ z\to\infty.
\end{equation}
This solution exists, it is unique, monotonous and positive \cite{Hastings_McLeod_1980}. 

\begin{remark} 
In \cite{Hastings_McLeod_1980}, the form of the full asymptotic expansion (\ref{pminf}) as $z\to-\infty$ was not established. For such $z$,
only the validity of the relation (\ref{HM-asymp}) was demonstrated there. The fact that the asymptotic expansion of $p_1(z)=y(z)$ as $z\to -\infty$ is determined precisely by the power expansion (\ref{pminf}) is verified as follows.

1) Let us assume that there are two different solutions $y_1(z)$ and $y_2(z)$ of P-II, which have asymptotics (\ref{HM-asymp}) as $z\to-\infty$. Their difference $W(z)=y_1(z)-y_2(z)$ satisfies the relation
\[
 W''(z)=S(z)W(z)=2\bigl(4z+y_1(z)^2+y_1(z)y_2(z)+y_2(z)^2\bigr)W(z).
\]
As $z\to-\infty$, we have
\[
S(z)=-16z(1+o(1))>\beta^2>0,
\]
for any positive constant $\beta$. Therefore, for such $z$, the function $W(z)$ cannot have either positive maxima or negative minima and, therefore, retains a constant sign for all sufficiently large $-z$. Therefore, for $z\to-\infty$ and an arbitrary $\beta>0$, either $W''>\beta^2W$ (if this difference is positive) or $-W''>\beta^2(-W)$. At the same time, by virtue of the validity of the asymptotics (\ref{HM-asymp}) as $z\to-\infty$, we have $|W(z)|=o(\exp{(-\beta z)})$. Therefore, it follows from \cite[Lemma 1.1]{Il'in_Suleimanov_2004} that for any two solutions $y_1(z)$ and $y_2(z)$ of P-II equation (\ref{P-II}), which have asymptotics (\ref{HM-asymp}) as $z\to-\infty$, the estimate holds
\begin{equation}\label{ambis}
 |y_1(z)-y_2(z)|=O(\exp{(\beta z)}) \qquad (\beta>0).
\end{equation}
 
2) The change of variables
\[
 \tau=(-z)^{-3},\qquad y(z)=(-z)^{1/2}\lambda(\tau)
\]
transforms equation (\ref{P-II}) into the ODE
\begin{equation}\label{Kodu}
 \tau^3\lambda''=-\tau^2\lambda'+\lambda\left(\frac{2\lambda^2-8}{9}+\frac{\tau}{36}\right),
\end{equation}
and the formal power solution (\ref{pminf}) of (\ref{P-II}) into the formal power solution of (\ref{Kodu})
\begin{equation}\label{forT}
\lambda(\tau)=2-\frac{\tau}{32}-\frac{73\tau^2}{2^{12}}-\frac{10657\tau^3}{2^{18}}+\dotsb. 
\end{equation}
According to the general theorem of Kuznetsov \cite{Kuznetsov_1972}, in a sufficiently small neighborhood of the point $\tau=0$ there also exists an infinitely differentiable solution of the ODE (\ref{Kodu}) for which the series (\ref{forT}) serves as the Taylor expansion. This means that for sufficiently large values of $-z$ there exists a solution $y_1(z)$ of P-II equation (\ref{P-II}) for which (\ref{pminf}) is the asymptotic expansion as $z\to-\infty$. Denoting by $y_2(z)$ the special Hastings--McLeod solution, we obtain, thanks to the estimate (\ref{ambis}) for large negative $z$, that this special solution also has the same power-law asymptotic expansion (\ref{pminf}).
\end{remark}

To find the next coefficients $p_2$ and $q_2$, we collect in (\ref{FGz-sys}) the terms with $t^{-2/3}$ and replace in them $q_1,q_2,q_3$ and $p''_1$ according to the relations (\ref{q12}), (\ref{q3}). This gives an expression for $q'_4$ in terms of $p_1,p_2,p_3$ and $p_4$, and an equation for $p_2$:
\[
 p'''_2 = 2(3p^2_1+4z)p'_2+4(3p_1p'_1+4)p_2+12zp_1p'_1.
\]
We will also use the constraint equations (\ref{constr}). Replacing the odd-numbered $u_n$ with the $F$ series and the even-numbered $u_n$ with the $G$ series, and collecting the coefficients at the powers of $t^{5/3}$, $t^{4/3}$, and $t$, we obtain exactly the relations (\ref{q12}) and (\ref{q3}). For the power $t^{2/3}$, an expression for $q_4$ arises, which we will not write out, and another equation for $p_2$:
\[
 p''_2= 2(3p^2_1+4z)p_2 +\frac{1}{2}((p'_1)^2-p^4_1+4zp^2_1+4p_1).
\]
Together with the previous one, it forms an overdetermined system with respect to $p_2$, and it is easy to verify that their only common solution is the function $p_2=(p^2_1+p'_1)/4$, where $p_1$ satisfies (\ref{q3}). After this, we find $q_2$ from (\ref{q12}) and, finally, arrive at the following statement.

\begin{proposition}\label{th:in-ser+}
FAS \textnormal{(\ref{aspf})} and \textnormal{(\ref{aspg})} as $t\to+\infty$ have the form
\begin{equation}\label{in-ser+}
 F=1+\frac{y(z)}{t^{1/3}}+\frac{y(z)^2+y'(z)}{4t^{2/3}}+O(t^{-1}),\qquad
 G=1-\frac{y(z)}{t^{1/3}}+\frac{y(z)^2-y'(z)}{4t^{2/3}}+O(t^{-1})
\end{equation} 
where $y(z)$ is the Hastings--McLeod solution.
\end{proposition}

Fig.~\ref{fig:HM} demonstrates the numerical solution $u_n(t)$ near the point $n=2t$, after passing to the stretched variables
\begin{equation}\label{nzxt}
 z=(x-1)t^{2/3},\qquad x=n/(2t),\qquad v_n(t)=t^{1/3}(u_n(t)-1).  
\end{equation}
Dots on the left and right graphs coincide. For $z<0$, they approximately lie on the branches of the parabola determined by the leading terms $f_0(x)$ and $g_0(x)$ of the original outer expansion (\ref{FGseries}) after the change (\ref{nzxt}): $v=2\sqrt{-z}-zt^{-1/3}$ for odd $n$ and $v=-2\sqrt{-z}-zt^{-1/3}$ for even $n$. The black curves correspond to the inner expansion (\ref{in-ser+}): the left graph shows approximations with leading terms $v=\pm y(z)$ only, the right graph shows approximations taking into account the next corrections, that is, $v=\pm y(z)+(y(z)^2\pm y'(z))/(4t^{1/3})$. As we see, the outer and inner expansions agree well for $z<0$.

\begin{remark}
Note that the transition of the solution from the parabola to the stationary state occurs in an interval whose width in the scale of the variable $z$ does not depend on $t$. Since the Hastings--McLeod solution is already close enough to 0 at $z=2$, one can take the characteristic width of the transition sector to be 2. Then, on the scale of the variable $x$, this width is equal to $2t^{-2/3}$ and tends to 0 as $t\to\infty$, while in the scale of the original variable $n$, the width grows as $4t^{1/3}$. Of course, the numerical factors here are rather arbitrary. Theoretically, according to the matching method, the width of this sector is estimated as $O(t^{-2/3+\xi})$, where $\xi$ is any number from the interval $(0,1/3)$.
\end{remark}

\begin{figure}[t]
\centerline{\includegraphics[width=0.95\textwidth]{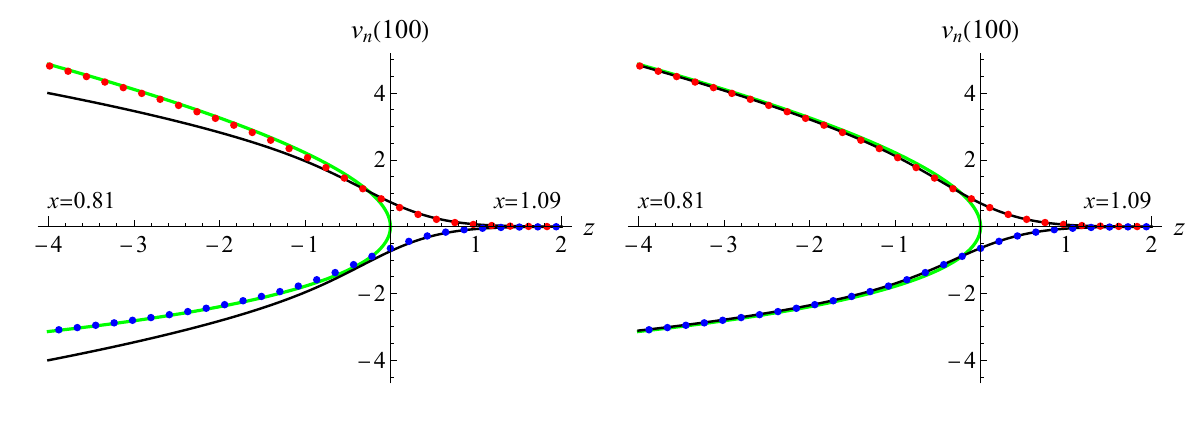}}
\caption{Solution near $n=2t$ in stretched variables $z$ and $v$, for $t=100$. Left: black curves are the graphs of the leading terms of the inner expansions. Right: black curves are the graphs of the leading terms plus first corrections.}
\label{fig:HM} 
\end{figure}

The asymptotic behavior of this solution of VL for $t\to-\infty$ is studied in a rather similar way. From Proposition \ref{th:alpha1-series} it follows that after stretching (\ref{zr}) for finite $z<0$ and $t\to-\infty$ the series for $F$ and $G$ coincide and their terms have the asymptotic behavior
\begin{alignat*}{2}
 f_0 &= 1+(-t)^{-2r}z,&
 \frac{f_1}{t} &= -(-t)^{-2r+3r-1}\frac{(1+O((-t)^{-r}))}{4\sqrt{-z}},\\
 \frac{f_2}{t^2} &= (-t)^{-2r+2(3r-1)}\frac{(1+O(t^{-r}))}{32(-z)^{2}},&
 \frac{f_3}{t^3} &= -(-t)^{-2r+3(3r-1)}\frac{5(1+O(t^{-r}))}{2^8(-z)^{7/2}},\\
 \frac{f_4}{t^4} &= (-t)^{-2r+4(3r-1)}\frac{9(1+O(t^{-r}))}{512(-z)^{5}},&
 \frac{f_5}{t^5} &= -(-t)^{-2r+5(3r-1)}\frac{803(1+O(t^{-r}))}{2^{15}(-z)^{13/2}},\\
 \frac{f_6}{t^6} &= (-t)^{-2r+6(3r-1)}\frac{1323(1+O(t^{-r}))}{2^{15}(-z)^{8}},\qquad &&\dotso.
\end{alignat*} 
The requirement that the orders with respect to $t$ should coincide leads, again, to the value $r=1/3$. In this case, unlike the expansions for $t\to+\infty$, the terms with even and odd numbers have the same order. Thus, the stretched variable for $t\to-\infty$ must also be determined by the formula (\ref{z}). With this choice of $r$ and $z$, the expansion for $F$ is transformed into the form
\begin{equation}\label{asmf}
 F=1+t^{-2/3}p_2(z)+t^{-1}p_3(z)+t^{-4/3}p_4(z)+\dotsb
\end{equation}
where $p_2(z)$ is the formal series
\begin{equation}\label{p2z}
 p_2(z)=z-\frac{1}{4(-z)^{1/2}}+\frac{1}{32(-z)^2}-\frac{5}{2^8(-z)^{7/2}}
  +\frac{9}{512(-z)^5}-\frac{803}{2^{15}(-z)^{13/2}}+\frac{1323}{2^{15}(-z)^8}+\dotsb.
\end{equation}
Similarly, the expansion for $G$ takes the form
\begin{equation}\label{asmg}
 G=1+t^{-2/3}q_2(z)+t^{-1}q_3(z)+t^{-4/3}q_4(z)+\dotsb
\end{equation}
where the coefficients $q_j$ are defined as $z\to-\infty$ by the same asymptotic series as $p_j$. As before, according to the matching method \cite{Il'in}, the asymptotics as $t\to-\infty$, suitable for all finite values of $z$, should be sought in the form of series (\ref{asmf}) and (\ref{asmg}), requiring that 1) the functions $p_2(z)$ and $q_2(z)$ expand in the power asymptotic series (\ref{p2z}) as $z\to-\infty$; 2) all coefficients $p_j(z)$ and $q_j(z)$ be rapidly decreasing as $z\to\infty$.

Substituting the series (\ref{asmf}) and (\ref{asmg}) into the system (\ref{FGz-sys}) with $v=-1$ leads, in the orders $t^{1/3}$, $t^0$ and $t^{-1/3}$, to the equations
\begin{gather*}
 p'_2=q'_2,\qquad q'_2=p'_2,\\
 p'_3=q'_3,\qquad q'_3=p'_3,\\
 p'_4+\frac{zp'_2+2p_2}{3}=q'_4+p_2q'_2+\frac{q'''_2}{24},\qquad 
 q'_4+\frac{zq'_2+2q_2}{3}=p'_4+q_2p'_2+\frac{p'''_2}{24}.
\end{gather*}
It is easy to see, taking into account that $p_j(z)$ and $q_j(z)$ vanish as $z\to\infty$, that these equations are equivalent to
\begin{equation}\label{pqid}
 p_2=q_2,\qquad p_3=q_3,\qquad p_4=q_4
\end{equation}
and the ODE for $p_2$ 
\[
 p'''_2+24p_2p'_2=8zp'_2+16p_2.
\]
Its order can be reduced by integrating with the factor $p_2-z$ and the well-known Miura transformation
\begin{equation}\label{Miura}
 4p_2=y'-y^2
\end{equation}
gives the expression of the sought solution in terms of the special Hastings--McLeod solution of P-II equation (\ref{P-II}), described above. As is easily seen, by virtue of the Miura transformation, the form of the explicitly written terms of the asymptotic expansion as $z\to-\infty$ (\ref{p2z}) follows from the form of the asymptotic expression (\ref{pminf}) of the Hastings--McLeod solution. As $z\to\infty$, the Miura transformation implies that $p_2(z)$ rapidly decreases, like the Hastings--McLeod solution. Thus, we arrive at the following statement.

\begin{proposition}\label{th:in-ser-}
FAS \textnormal{(\ref{asmf})}, \textnormal{(\ref{asmg})} as $t\to-\infty$ are of the form
\begin{equation}\label{int-ser-}
 F=1+\frac{y'(z)-y(z)^2}{4t^{2/3}}+O(t^{-1}),\qquad G=1+\frac{y'(z)-y(z)^2}{4t^{2/3}}+O(t^{-1})
\end{equation} 
where $y(z)$ is the Hastings--McLeod solution.
\end{proposition}

\begin{remark}
The Hastings--McLeod special solution is widely used in mathematical physics for description of various transient regimes, including transition regions of asymptotic behavior with respect to the discrete independent variable of difference equations \cite{Novokshenov_2022}. Currently, the properties of this special function are well described, and its behavior has been studied numerically. In particular, in our numerical calculations presented in Fig.~\ref{fig:HM}, we used the same program for computing the Hastings--McLeod solution as in \cite{Novokshenov_2009}.
\end{remark}

\section{Correction terms in the case \texorpdfstring{$\alpha\ne1$}{alpha != 1}}\label{s:ab} 

The case of arbitrary $\alpha$ can, in principle, also be studied using a suitable constraint consistent with the dynamics in $t$. However, such a constraint is more complicated than (\ref{constr}) (see \cite[Sec.~5]{Adler_2024}) and leads to rather involved calculations. It is simpler to use the well-known substitution between the Volterra lattice and the Toda lattice in the Flaschka variables \cite{Flaschka_1974}, which allows us to reduce the problem to the case $\alpha=1$. Recall that this substitution reads
\begin{equation}\label{ab}
 a_k=u_{2k-1}+u_{2k},\qquad b_k=u_{2k}u_{2k+1}
\end{equation}
and it transforms the lattice (\ref{VL}) to
\begin{equation}\label{TL}
 \dot a_k=b_k-b_{k-1},\qquad \dot b_k= b_k(a_{k+1}-a_k),
\end{equation}
while the initial-boundary conditions (\ref{VL0}) turn into
\begin{equation}\label{ab0}
 b_0(t)=0,\quad a_1(0)=a_2(0)=\dots=\alpha+\alpha^{-1},\qquad b_1(0)=b_2(0)=\dots=1.
\end{equation}
Note that equations (\ref{TL}) are invariant with respect to the transformation 
\begin{equation}\label{shift}
 \tilde a_k=a_k+\tilde\alpha+\tilde\alpha^{-1}-\alpha-\alpha^{-1},\qquad
 \tilde b_k=b_k
\end{equation}
which does not change the form of the initial data and is equivalent just to choosing a different value of $\alpha$. In other words, solutions of the Volterra lattice corresponding to different $\alpha$ are mapped by the substitution (\ref{ab}) into solutions of the Toda lattice related by a trivial transformation. This property can be used to recalculate the coefficients of the asymptotic series. Moreover, equations (\ref{TL}) and (\ref{ab0}) are also invariant under the transformation
\begin{equation}\label{refl}
 \hat a_k(t) = 2(\alpha+\alpha^{-1})-a_k(-t),\qquad 
 \hat b_k(t)\mapsto b_k(-t),
\end{equation}
which relates the asymptotic series for $t\to+\infty$ and $t\to-\infty$.

In principle, the recalculation can be done without explicitly using the variables $a_k$ and $b_k$, by substituting the asymptotic expansions for $u_n$ and $\tilde u_n$ into the relations
\[
 \tilde u_{2k-1}+\tilde u_{2k} = 
  u_{2k-1}+u_{2k}+\tilde\alpha+\tilde\alpha^{-1}-\alpha-\alpha^{-1},\qquad 
 \tilde u_{2k}\tilde u_{2k+1}=u_{2k}u_{2k+1},
\]
which define a kind of B\"acklund transformation for VL. However, it is somewhat simpler to pre-calculate the expansions for $a_k$ and $b_k$, which are also of independent interest since they describe the asymptotic behavior of the solution to the initial-boundary value problem (\ref{TL}), (\ref{ab0}). So, we set
\begin{equation}\label{ab-series}
\begin{aligned}
 a_k(t)&=A(x,t)=A_0(x)+\frac{A_1(x)}{t}+\frac{A_2(x)}{t^2}+\dotsb,\\
 b_k(t)&=B(x,t)=B_0(x)+\frac{B_1(x)}{t}+\frac{B_2(x)}{t^2}+\dotsb
\end{aligned}
\end{equation}
where $x=k/(vt)$ (this formula agrees with (\ref{x}), since we are now essentially assigning both variables $a_k$ and $b_k$ to the same node $n=2k$ of the old lattice). As in the case of series (\ref{FGseries}), the equations for the coefficients arising from substitution into (\ref{TL}) are underdetermined. In particular, at zeroth order we have the system
\[
 B'_0+vxA'_0=0,\quad vxP'_0+B_0A'_0=0,
\]
which implies
\[
 B_0=v^2x^2,\quad A_0=\operatorname{const}-2vx;
\]
first-order equations are reduced to a single relation $B_1=vx(1-A_1)$, that is, the function $A_1$ is arbitrary; second-order equations allow us to express $A_2$ and $B_2$ through $A_1$ and a single integration constant. However, since the series corresponding to the variables $u_n$ in the case of $\alpha=1$ are already known, we can calculate the series $A$ and $B$ from them, according to the formulas
\begin{equation}\label{ABFG}
 A(x,t)= F\Bigl(x-\frac{1}{2vt},t\Bigr)+G(x,t),\qquad 
 B(x,t)= G(x,t)F\Bigl(x+\frac{1}{2vt},t\Bigr),
\end{equation}
which leads to the relations
\begin{equation}\label{ABfg}
\begin{gathered}
 A_0 = f_0+g_0, \qquad 
 B_0 = f_0g_0, \\
 A_1 = f_1+g_1-\frac{f'_0}{2v}, \qquad
 B_1 = g_0f_1+f_0g_1+\frac{f'_0g_0}{2v},\\
 A_2 = f_2+g_2-\frac{f'_1}{2v}+\frac{f''_0}{8v^2}, \qquad 
 B_2 = g_0f_2+f_0g_2+f_1g_1+\frac{f'_1g_0+f'_0g_1}{2v}+\frac{f''_0g_0}{8v^2},
\end{gathered}
\end{equation}
and so on. We choose as $f_j$ and $g_j$ the functions found in Section \ref{s:a1} for the case $\alpha=1$ and $t\to+\infty$ (with $v=1$) and calculate the coefficients $A_j$ and $B_j$ from them. Then, applying the transformations (\ref{shift}) and (\ref{refl}), we arrive at the following formulas (as always, $\sqrt{1-x}$ denotes the positive value of the root).

\begin{proposition}\label{th:ab-series}
For the Toda lattice solution with initial data \textnormal{(\ref{ab0})}, the first few coefficients of the asymptotic series \textnormal{(\ref{ab-series})} as $t\to\pm\infty$ are equal to
\begin{gather*}
 A_0= \alpha+\alpha^{-1}\pm2(1-x),\quad 
 A_1= \frac{1}{2\sqrt{1-x}},\quad 
 A_2= \mp\frac{1}{16(1-x)^2},\quad 
 A_3= \frac{(4+x)\sqrt{1-x}}{128(1-x)^4},\\
 B_0= x^2,\qquad 
 B_1= \pm x \mp\frac{x}{2\sqrt{1-x}},\qquad
 B_2= \frac{4-6x+3x^2}{16(1-x)^2}-\frac{(2-x)\sqrt{1-x}}{8(1-x)^2},\\
 B_3= \pm\frac{1}{16(1-x)^3} \mp\frac{(8-4x+x^2)\sqrt{1-x}}{128(1-x)^4}.
\end{gather*}
\end{proposition}

Now, to find the coefficients $f_j$ and $g_j$ corresponding to an arbitrary value of $\alpha$, we need to invert the mapping (\ref{ABfg}) between the coefficients of the series. It is easy to see that the first pair of equations coincides with (\ref{fg1}) and (\ref{fg2}), that is, they define the same functions $f_0$ and $g_0$ as in Proposition \ref{th:fg-series}, up to the sign of the root, which we have established by other reasoning. The subsequent equations (\ref{ABfg}) have the form
\[
 A_j=f_j+g_j+\dotsb,\qquad B_j=g_0f_j+f_0g_j+\dotsb,
\]
where the dots denote terms containing only $f_i$ and $g_i$ with $i<j$. These equations are uniquely solvable with respect to $f_j$ and $g_j$, since $f_0-g_0\ne0$ (except for $\alpha=1$ and $t<0$, but for this case the answer is already known). Using this scheme, one can, in principle, find an arbitrary number of coefficients. Solving the equations for $j=1$, we find the functions (\ref{fg1p}) and (\ref{fg1m}), which completes the proof of Proposition \ref{th:fg-series}. The expressions for the next coefficients $f_2$ and $g_2$ are already quite cumbersome.

\begin{remark}
The substitution
\begin{equation}\label{ab'}
 a_k=u_{2k-2}+u_{2k-1},\qquad b_k=u_{2k-1}u_{2k},
\end{equation}
which differs from (\ref{ab}) only by changing the index of $u_n$, also brings to the lattice equations (\ref{TL}) with boundary condition $b_0(t)=0$, but with slightly different initial data:
\[
 a_1(0)=\alpha,\qquad a_2(0)=a_3(0)=\dots=\alpha+\alpha^{-1},\qquad b_1(0)=b_2(0)=\dots=1.
\]
It is easy to see that in this case the shift $a_k\to a_k+\operatorname{const}$ can no longer be compensated by changing $\alpha$, so this substitution is not suitable for recalculating the series. However, it can be used for generating new solutions while preserving the boundary condition. In particular, note that under the substitution (\ref{ab}), the bipartite kink (\ref{VL_kink}) maps to the constant solution $b_k=1$, $a_k=\alpha+\alpha^{-1}$, while the substitution (\ref{ab'}) transforms it to a kink of the Toda lattice.
\end{remark}

\section{Asymptotics near the point \texorpdfstring{$x=1$}{x=1} for \texorpdfstring{$\alpha\neq 1$}{alpha != 1}}\label{s:x1}

The transformation to the Toda lattice allows us to calculate for an arbitrary $\alpha$ the asymptotics of the solution near the point $x=1$, obtained in Section \ref{s:x1a1}. First, we write out the asymptotics for the solution of the Toda lattice itself. Relations (\ref{ABFG})
\begin{align*}
 A(x,t)&= F\Bigl(x-\frac{1}{2vt},t\Bigr)+G(x,t)
  = G(x,t) +F(x,t)-\frac{F_x(x,t)}{2vt}+\frac{F_{xx}(x,t)}{8vt^2}+\dotsb,\\
 B(x,t)&= G(x,t)F\Bigl(x+\frac{1}{2vt},t\Bigr)
  = G(x,t)\left(F(x,t)+\frac{F_x(x,t)}{2vt}+\frac{F_{xx}(x,t)}{8vt^2}+\dotsb\right)
\end{align*}
take the following form after passing to the stretched variable (\ref{z}) $z=(x-1)t^{2/3}$:
\begin{equation}\label{ABz}
 A= G +F-\frac{F_z}{2vt^{1/3}}+\frac{F_{zz}}{8vt^{2/3}}+\dotsb,\qquad
 B= G\left(F+\frac{F_z}{2vt^{1/3}}+\frac{F_{zz}}{8vt^{2/3}}+\dotsb\right).
\end{equation}
For $t>0$, we choose $v=1$ and substitute expansions (\ref{aspf}) and (\ref{aspg}) instead of $F$ and $G$, which gives
\[
 A=2+\frac{p_1+q_1}{t^{1/3}}+\frac{p_2+q_2-p'_1/2}{t^{2/3}}+\dotsb,\qquad
 B=1+\frac{p_1+q_1}{t^{1/3}}+\frac{p_2+q_2+p_1q_1+p'_1/2}{t^{2/3}}+\dotsb.
\]
From here it follows, by virtue of the identities (\ref{q12}), that for $\alpha=1$ the inner expansion for $t\to+\infty$, valid for all finite $z$, has the form
\[
 A=2-\frac{y'-y^2}{2t^{2/3}}+\dotsb,\qquad B=1+\frac{y'-y^2}{2t^{2/3}}+\dotsb
\]
where $y(z)$ is the special Hastings--McLeod solution of P-II equation (\ref{P-II}). For arbitrary $\alpha$, it suffices to apply the transformation (\ref{shift}), which yields
\begin{equation}\label{ABy+}
 A=\alpha+\alpha^{-1}-\frac{y'-y^2}{2t^{2/3}}+\dotsb,\qquad
 B=1+\frac{y'-y^2}{2t^{2/3}}+\dotsb,\qquad t\to+\infty.
\end{equation}
Similarly, for $t<0$ we choose $v=-1$ in (\ref{ABz}) and replace $F$ and $G$ with the expansions (\ref{asmf}) and (\ref{asmg}), which gives
\[
 A=2+\frac{p_2+q_2}{t^{2/3}}+\dotsb,\qquad B=1+\frac{p_2+q_2}{t^{2/3}}+\dotsb.
\]
Applying the shift (\ref{shift}) and taking into account the identities (\ref{pqid}) and (\ref{Miura}), we obtain an inner expansion as $t\to-\infty$, suitable for all finite $z$ and any $\alpha$:
\begin{equation}\label{ABy-}
 A=\alpha+\alpha^{-1}+\frac{y'-y^2}{2t^{2/3}}+\dotsb,\qquad B=1+\frac{y'-y^2}{2t^{2/3}}+\dotsb,\qquad t\to-\infty,
\end{equation}
where $y$ is the Hastings--McLeod solution.

Now we need to return back to the Volterra lattice. To do this, we will use the results of Proposition \ref{th:fg-series}. From these it follows that after passing to the stretched variable (\ref{z}), the expansions (\ref{FGseries}) for $t\to\infty$ take the form
\begin{align*}
 F&= \mu+\frac{\mu}{(\mu-1)t^{2/3}}
  \left(-2z+\frac{1}{2(-z)^{1/2}}+\dotsb\right)+\dotsb, \\ 
 G&= \frac{1}{\mu}+\frac{1}{(\mu-1)t^{2/3}}
  \left(2z-\frac{1}{2(-z)^{1/2}}+\dotsb\right)+\dotsb
\end{align*}
where $\mu=\max(\alpha,\alpha^{-1})$, and the expansions for $t\to-\infty$ have the form
\begin{align*}
 F&= \alpha+\frac{\alpha}{(\alpha+1)t^{2/3}}
  \left(2z+\frac{1}{2(-z)^{1/2}}+\dotsb\right)+\dotsb,\\ 
 G&= \frac{1}{\alpha}+\frac{1}{(\alpha+1)t^{2/3}}
  \left(2z+\frac{1}{2(-z)^{1/2}}+\dotsb\right)+\dotsb.
\end{align*}
Therefore, the asymptotic solutions $F$ and $G$, valid for all finite $z$, should be written as the following series. For $t\to+\infty$:
\begin{equation}\label{FG-in+}
 F=\mu+\frac{\psi_2(z)}{t^{2/3}} +\frac{\psi_3(z)}{t}+\dotsb, \qquad
 G=\frac{1}{\mu}+\frac{\varphi_2(z)}{t^{2/3}} +\frac{\varphi_3(z)}{t}+\dotsb,
\end{equation}
where all coefficients rapidly tend to zero as $z\to\infty$, and the coefficients $\psi_2(z)$ and $\varphi_2(z)$ are expanded, as $z\to-\infty$, into the power asymptotic series 
\begin{equation}\label{psi+z}
 \psi_2(z)=\frac{\mu}{\mu-1}\left(-2z+\frac{1}{2(-z)^{1/2}}+\dotsb\right),\qquad 
 \varphi_2(z)=\frac{1}{\mu-1}\left(2z-\frac{1}{2(-z)^{1/2}}+\dotsb\right);
\end{equation}
for $t\to-\infty$:
\begin{equation}\label{FG-in-}
 F=\alpha+\frac{\psi_2(z)}{t^{2/3}} +\frac{\psi_3(z)}{t}+\dotsb, \qquad
 G=\frac{1}{\alpha}+\frac{\varphi_2(z)}{t^{2/3}} +\frac{\varphi_3(z)}{t}+\dotsb,
\end{equation}
where all coefficients rapidly tend to zero as $z\to\infty$, and the coefficients $\psi_2(z)$ and $\varphi_2(z)$ are expanded, as $z\to-\infty$, into the power asymptotic series 
\begin{equation}\label{psi-z}
 \psi_2(z)=\frac{\alpha}{\alpha+1}\left(2z+\frac{1}{2(-z)^{1/2}}+\dotsb\right),\qquad 
 \varphi_2(z)=\frac{1}{\alpha+1}\left(2z+\frac{1}{2(-z)^{1/2}}+\dotsb\right). 
\end{equation}
Direct substitution of the expansions (\ref{FG-in+}) and (\ref{FG-in-}) into the relations (\ref{ABz}) gives
\begin{gather*}
 A=\mu+\mu^{-1}+\frac{\psi_2+\varphi_2}{t^{2/3}}+\dotsb,\quad 
 B=1+\frac{\mu^{-1}\psi_2+\mu\varphi_2}{t^{2/3}}+\dotsb,\quad t\to+\infty,\\
 A=\alpha+\alpha^{-1}+\frac{\psi_2+\varphi_2}{t^{2/3}}+\dotsb,\quad 
 B=1+\frac{\alpha^{-1}\psi_2+\alpha\varphi_2}{t^{2/3}}+\dotsb,\quad t\to-\infty.
\end{gather*}
By comparing these relations with (\ref{ABy+}) and (\ref{ABy-}) we obtain the following systems of linear equations for $\psi_2$ and $\varphi_2$:
\begin{alignat*}{3}
 &\psi_2+\varphi_2=-\frac{y'-y^2}{2},\qquad&  
 &\mu^{-1}\psi_2+\mu\varphi_2=\frac{y'-y^2}{2} \qquad& 
 &(\text{for}~ t\to \infty);\\
 &\psi_2+\varphi_2=\frac{y'-y^2}{2}, &
 &\alpha^{-1}\psi_2+\alpha\varphi_2=\frac{y'-y^2}{2} &
 &(\text{for}~ t\to-\infty).
\end{alignat*}
Solving these linear systems, we express $\psi_2$ and $\varphi_2$ in terms of the function $y'-y^2$. From the obtained formulas and the form of the power asymptotics (\ref{pminf}) of the Hastings--McLeod solution $y(z)$, it follows that the coefficients $\psi_2$ and $\varphi_2$ indeed have asymptotics of the form (\ref{psi+z}), (\ref{psi-z}) as $z\to-\infty$. As $z\to\infty$, these coefficients rapidly decrease, similar to the Hastings--McLeod solution. Finally, we arrive at the following proposition, which complements Propositions \ref{th:in-ser+} and \ref{th:in-ser-}.

\begin{proposition}\label{th:in-ser-alpha} 
For $\alpha\ne1$, the formal asymptotic solutions $F$ and $G$ for $t\to\pm\infty$, suitable for all finite values of $z=(x-1)t^{2/3}$ have the form
\begin{align*}
 & F=\mu-\frac{\mu(y'-y^2)}{2(\mu-1)t^{2/3}} +O(t^{-1}),\qquad 
   G=\frac{1}{\mu}+\frac{y'-y^2}{2(\mu-1)t^{2/3}} +O(t^{-1}),\qquad t\to+\infty,\\
 & F=\alpha+\frac{\alpha(y'-y^2)}{2(\alpha+1)t^{2/3}} +O(t^{-1}),\qquad 
   G=\frac{1}{\alpha}+\frac{y'-y^2}{2(\alpha+1)t^{2/3}} +O(t^{-1}),\qquad t\to-\infty
\end{align*}
where $\mu=\max(\alpha,\alpha^{-1})$ and $y(z)$ is the Hastings--McLeod solution.
\end{proposition}
 
\section{Uniform asymptotics of one regular solution of the fifth Painlev\'e equation}\label{s:P5} 

For $\alpha=1$, according to \cite[Proposition 2]{Adler_Shabat_2019}, the function
\begin{equation}\label{shabdN1}
\rho(t) =\rho_n(t)=1-\frac{8t}{2t(u_{2n}+u_{2n+1})+4n+2}
\end{equation} 
satisfies the fifth Painlev\'e equation
\begin{equation}\label{painf}
 \rho''=\left(\frac{1}{2\rho}+\frac{1}{\rho-1}\right)(\rho')^2-\frac{\rho'}{t}
  +\frac{(\rho-1)^2}{8t^2}\left(\rho-\frac{1}{\rho}\right)
  -\frac{2(4n+2)\rho}{t}-\frac{8\rho(\rho+1)}{\rho-1}.
\end{equation}
This solution is regular at the origin:
\[
 \rho(t)=1-\frac{4t}{2n+1}+O(t^2),\qquad t\to 0.
\]
As a simple application of results obtained in previous sections, we describe here, by passing from $n$ to the independent variable $x=n/|t|$, the asymptotics of this solution as $t\to\pm\infty$, which is uniform for $n\geqslant 0$. According to Proposition \ref{th:alpha1-series}, for $0\leqslant x<1$ and $t\to\infty$
\begin{gather*}
 2t(u_{2n+1} +u_{2n})+4xt+2=2t(f_0(x)+g_0(x)+2x)+2(f_1(x)+g_1(x)+1)+f'_0(x)+O(t^{-1})\\
 =8t-\frac{1}{(1-x)^{1/2}}+O(t^{-1}).
\end{gather*} 
From the representation (\ref{shabdN1}) and this formula it follows that 
\begin{equation}\label{outlp}
 \rho(t)=-\frac{1}{8t(1-x)^{1/2}}+O(t^{-2}),
\end{equation} 
as $t\to\infty$ and for non-negative values of $x$ to the left of a small neighborhood of $x=1$. To the right of this neighborhood, from the fact that the solution of the Volterra lattice tends to the solution $u_k(t)=1$ and from the representation (\ref{shabdN1}), it follows that this solution of P-V equation (\ref{painf}) has the asymptotics
\begin{equation}\label{outrp}
 \rho(t)=\frac{x-1}{x+1}+O(t^{-1}).
\end{equation}
According to Proposition \ref{th:in-ser+}, we have
\[
 2t(u_{2n+1} +u_{2n})+4xt+2=8t+\frac{y'_z+y^2+4z}{t^{1/3}}+O(t^{-1})
\]
for finite values of the stretched variable $z=(x-1)t^{-2/3}$, where $y(z)$ is the special Hastings--McLeod solution of P-II equation (\ref{P-II}), which exponentially tends to zero as $z\to\infty$ and has a power-law asymptotic expansion (\ref{pminf}) as $z\to-\infty$. Therefore, the representation (\ref{shabdN1}) implies that for finite $z$ and $t\to\infty$ the asymptotic expansion of this solution of P-V equation has the form
\begin{equation}\label{promp}
 \rho(t)=\frac{y'_z+y^2+4z}{8t^{2/3}}+O(t^{-1}).
\end{equation}
From the asymptotic behavior of $y(z)$ as $z\to\pm\infty$ it is clear that the leading term of this inner expansion is consistent with both parts of the outer expansion, which are described by formulas (\ref{outlp}) and (\ref{outrp}) outside a neighborhood of the point $x=1$ which is small as $t\to\infty$. Indeed, 
\[
 -\frac{1}{8t(1-x)^{1/2}}=-\frac{1}{8t^{2/3}(-z)^{1/2}}
\]
and
\[
 \frac{x-1}{x+1}=\frac{z}{2t^{2/3}}+O(t^{-2/3}).
\]

We now turn to the description of the uniform in $n$ asymptotics of $\rho(t)$ as $t\to-\infty$. The reasoning and calculations here are basically the same as in the previous paragraph. Given that for negative $t$ we have $n=-xt$, the following relation holds to the left of the point $x=1$, according to Proposition \ref{th:alpha1-series}:
\begin{gather*}
 2t(u_{2n+1} +u_{2n})+4n+2=2t(f_0(x)+g_0(x)-2x)+2(f_1(x)+g_1(x)+1)-f'_0(x)+O(t^{-1})\\
 =\frac{1}{(1-x)^{1/2}}+O(t^{-1}).
\end{gather*}
Therefore, the representation (\ref{shabdN1}) implies that for $t\to-\infty$ and to the left of the point $x=1$ outside its small neighborhood, the asymptotics is 
\begin{equation}\label{outlm}
 \rho=-8t(1-x)^{1/2}+O(1).
\end{equation}
To the right of this small neighborhood, from the representation (\ref{shabdN1}) and the fact that the solution of the Volterra lattice tends to $u_k(t)=1$, the validity of the asymptotics follows
\begin{equation}\label{outrm}
 \rho(t)=\frac{x+1}{x-1}+O(t^{-1}).
\end{equation}
In this very neighborhood, according to the representation (\ref{shabdN1}) and Proposition \ref{th:in-ser-}, the asymptotics is valid for $t\to-\infty$
\begin{equation}\label{promm}
 \rho(t)=-\frac{8t^{2/3}}{y'_z-y^2-4z}+O(t^{1/3}),\qquad (z=(x-1)t^{2/3}), 
\end{equation}
where $y(z)$ is the Hastings--McLeod solution. From the asymptotic behavior of $y(z)$ and the relations (for negative $t$, the equality $(1-x)^{1/2}=-t^{1/3}(-z)^{1/2}$ holds)
\[
 -8t(1-x)^{1/2}=8t^{2/3}(-z)^{1/2},\qquad \frac{x+1}{x-1}=\frac{2t^{2/3}}{z} +O(1),
\]
we conclude that the leading term of the inner asymptotic expansion (\ref{promm}) is consistent with both parts (\ref{outlm}) and (\ref{outrm}) of the outer asymptotic expansion of $\rho(t)$.

\section{Shock waves in the Volterra and Toda lattices}\label{s:shock} 

Here we present numerical results in the case of a more general, non-zero boundary condition at $n=0$:
\begin{equation}\label{VLshock}
\begin{gathered}
 \dot u_n=u_n(u_{n+1}-u_{n-1}),\qquad n=1,2,\dotsc,\\
 u_0(t)=\beta,\qquad u_1(0)=u_3(0)=\dotsb=\alpha,\qquad u_2(0)=u_4(0)=\dotsb=1/\alpha
\end{gathered}
\end{equation}
where $\alpha>0$ and $\beta\ge0$ are constant parameters. The behavior of solutions for $\beta\ne0$ is more complex, but for some values of $\beta$ they still contain a self-similar region described by the ansatz (\ref{x}), (\ref{uFG}). The formulas for the leading coefficients $f_0(x)$ and $g_0(x)$ of the asymptotics given in Proposition \ref{th:fg-series} remain unchanged, since their derivation was fairly universal. Only the interval on which this asymptotics works has changed: now it may not be the entire segment $x\in[0,1]$, but a shorter one $[x_0,1]$. For some values of $\beta$, the self-similar region is completely absent. The solution regime is determined by fairly simple geometry: roughly speaking, the point $x_0$ is the abscissa of the intersection of the line $u=\beta$ with one of the branches of the limit parabola defined by the graphs $u=f_0(x)$ and $u=g_0(x)$. In general, comparison with the solutions presented in the previous sections already gives some idea of the solutions to the problem (\ref{VLshock}), although their detailed description requires further study. 

\begin{figure}[t]
\centerline{\includegraphics[width=0.95\textwidth]{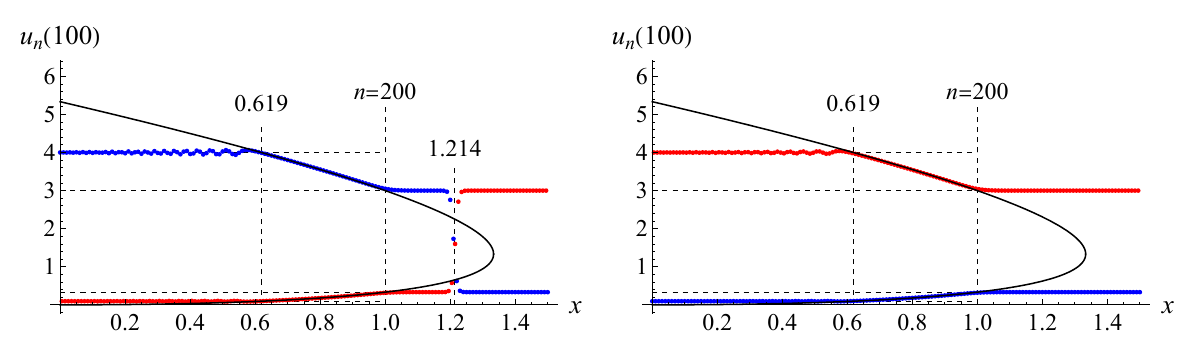}}
\caption{Solutions of the problem (\ref{VLshock}) for $t=100$, $\alpha=3$ and $\beta=4$ (left), $\beta=0.0957$ (right).}
\label{fig:sh1} 
\end{figure}

\begin{figure}[t]
\centerline{\includegraphics[width=0.95\textwidth]{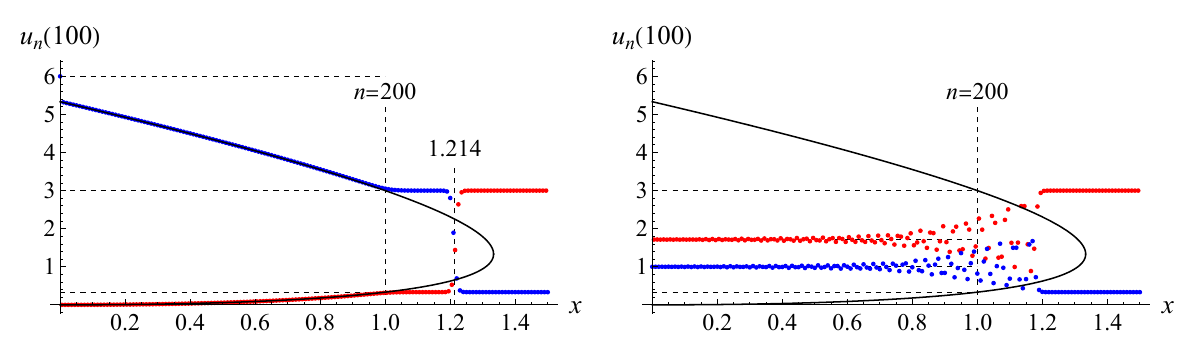}}
\caption{Solutions of the problem (\ref{VLshock}) for $t=100$, $\alpha=3$ and $\beta=6$ (left), $\beta=1$ (right).}
\label{fig:sh2} 
\end{figure}

We first consider the evolution for $t>0$. In this case, the parabola is defined by equations
\[
 f_0g_0=x^2,\qquad f_0+g_0=4c-2x,\qquad 4c=\alpha+1/\alpha+2,
\]
and the abscissa of its intersection with the line $u=\beta\ge0$ is
\[
 x_0= 2\sqrt{c\beta}-\beta
\]
(the second solution $x_0=-2\sqrt{c\beta}-\beta$ is obviously negative and is of no interest). The same point $x_0$ corresponds to another value $\beta$, equal to
\[
 \tilde\beta=4c+\beta-4\sqrt{c\beta}.
\]
If $\beta$ lies in one of the intervals $[0,1/\mu]$ or $[\mu,4c]$, where $\mu=\max(\alpha,1/\alpha)$, then $\tilde\beta$ lies in the other one and $x_0<1$. Solutions corresponding to the values
\[
 \mu=\alpha=3,\qquad \beta=4,\qquad \tilde\beta\approx 0.0957,\qquad x_0\approx 0.619
\]
are shown in Fig.~\ref{fig:sh1}, where dashed horizontal lines correspond to the ordinates $\mu$, $1/\mu$, $\beta$, and $\tilde\beta$ (cf. Fig.~\ref{fig:gr3}, which corresponds to the same value of $\alpha$). On the interval from $0$ to $x_0$, a certain transition regime is observed; to the right of $x_0$, the solution demonstrate the familiar self-similar behavior, and near $x=1$ it approaches the initial steady state. Replacing the boundary condition $u_0(t)=\beta$ with $u_0(t)=\tilde\beta$ effectively leads to a reversal of the roles of variables with even and odd $n$, as well as replacing $\alpha$ with $1/\alpha$. As a result, in one of the solutions, a kink is formed at the point $x_*>1$, which was described in Section \ref{s:leading}.

If $\beta\ge4c$ then $x_0\le0$, and the solution quickly tends to the self-similar regime on the entire interval $[0,1]$, that is, essentially, this solution behaves the same as under the zero boundary condition $\beta=0$ (Fig.~\ref{fig:sh2} on the left). Finally, if $1/\mu<\beta<\mu$, then $x_0>1$, and there is no self-similar region. In this case, the solution is not a decay wave, but a faster shock wave (Fig.~\ref{fig:sh2} on the right). The left diagram in Fig.~\ref{fig:diagrams} shows the regions in the quadrant $\alpha>0$ and $\beta\ge0$ corresponding to the described types of solutions.

\begin{figure}[t]
\centerline{\includegraphics[width=0.38\textwidth]{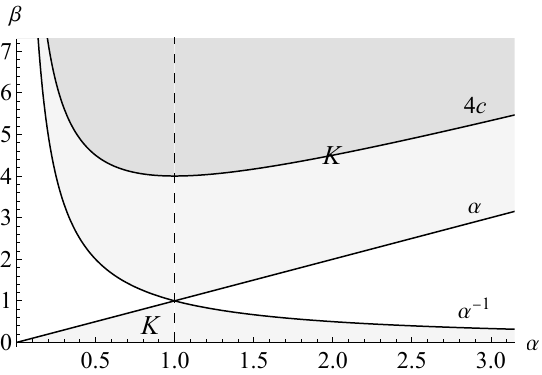}\qquad\qquad
\includegraphics[width=0.38\textwidth]{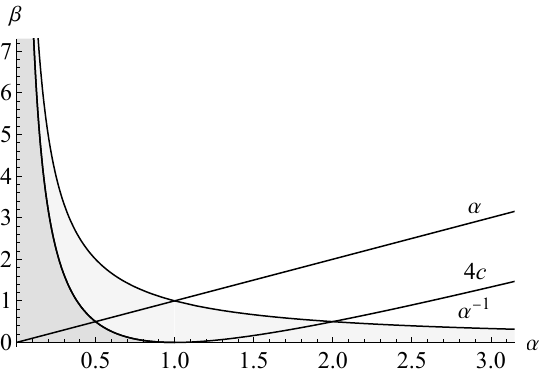}}
\caption{Regions with different solution types depending on $\alpha$ and $\beta$. Left: $t>0$, right: $t<0$. White color: shock waves, light gray: self-similar regime on part of the interval $[0,1]$, dark gray: self-similar regime on the entire interval $[0,1]$ (the $\beta=0$ axis also belongs to this region). Letter $K$ marks regions in which the solution has a reversal at the point $x_*$.}\label{fig:diagrams} 
\end{figure}

For $t<0$, the limit parabola is defined by equations
\[
 f_0g_0=x^2,\qquad f_0+g_0=4c+2x,\qquad 4c=\alpha+1/\alpha-2.
\]
If $\alpha<1$ then the solution is asymptotically self-similar on the entire interval $x\in[0,1]$ for all $\beta\in[0,4c]$; for $\beta\in[4c,1/\alpha]$ there is a self-similar sector on the interval $[x_0,1]$, where $x_0=\beta-2\sqrt{c\beta}$; for $\beta>1/\alpha$ the point $x_0$ lies to the right of $1$ and the solution behaves like a shock wave. If $\alpha>1$, then for $\beta\in[0,1/\alpha]$ a self-similar sector $[x_0,1]$ is observed, where $x_0=\beta+2\sqrt{c\beta}$, and for $\beta>1/\alpha$ we again have a shock wave. These types of solutions are illustrated by the right diagram in Fig.~\ref{fig:diagrams}. Fig.~\ref{fig:sh3} shows the solutions corresponding to $\alpha=1/3$ and $\alpha=3$.

\begin{figure}[t]
\centerline{\includegraphics[width=0.95\textwidth]{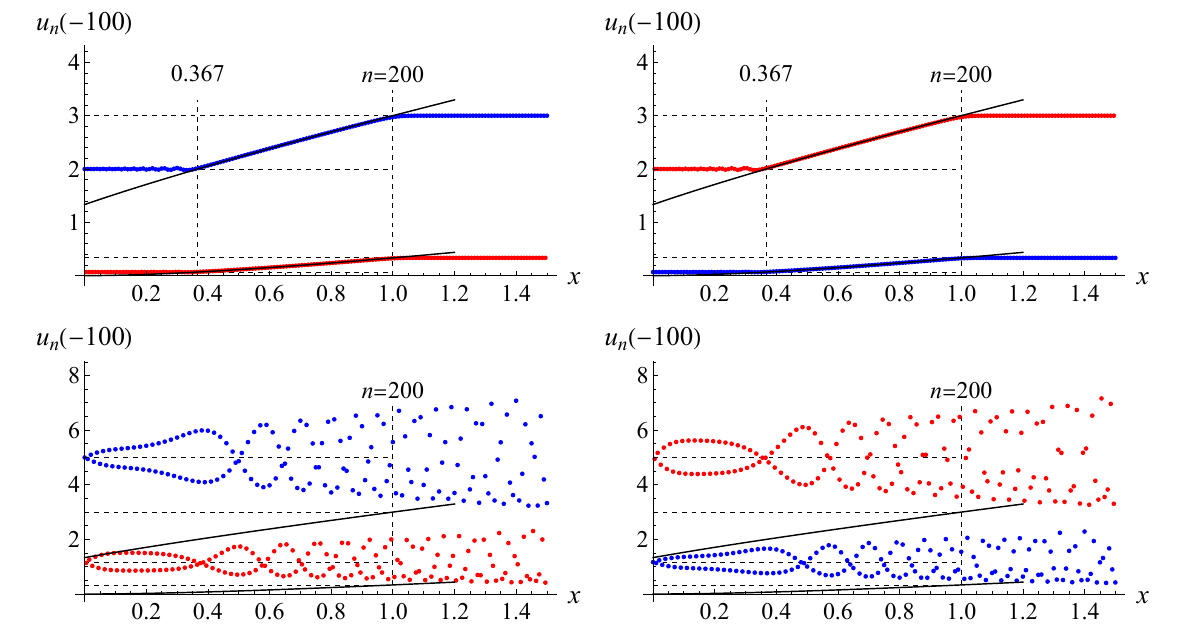}}
\caption{Solutions of the problem (\ref{VLshock}) for $t=-100$. Left: $\alpha=1/3$ and $\beta=2$ (top), $\beta=5$ (bottom). Right: $\alpha=3$ and $\beta=0.0673$ (top), $\beta=1.169$ (bottom).}
\label{fig:sh3} 
\end{figure}

To conclude this section, let us compare (\ref{VLshock}) with the Toda lattice shock wave problem studied in \cite{Holian_Straub_1978, Holian_Flaschka_McLaughlin_1981, Venakides_Deift_Oba_1991, Kamvissis_1993, Deift_Kamvissis_Kriecherbauer_Zhou_1996}. In the Flaschka variables, it is written as the initial-boundary value problem
\begin{gather}
\label{TL_k}
 \dot a_k=b_k-b_{k-1},\quad k=1,2,3,\dotsc,\qquad 
 \dot b_k=b_k(a_{k+1}-a_k),\quad k=0,1,2,\dotsc,\\
\label{TL_bc1}
 a_0(t)=0,\qquad a_1(0)=a_2(0)=\dotsb=a,\qquad b_0(0)=b_1(0)=\dotsb=1.
\end{gather}
Continuation by symmetry $a_{-k}=-a_k$, $b_{-k}=b_{k-1}$ on the entire $k$ axis brings to the Cauchy problem with step-like initial data. To solve it, in \cite{Venakides_Deift_Oba_1991, Kamvissis_1993, Deift_Kamvissis_Kriecherbauer_Zhou_1996} the inverse scattering method was developed, which made it possible to consider more general initial data tending to constants as $|k|\to\infty$. The substitution 
\[
 a_k=a-\dot q_k,\qquad b_k=e^{q_k-q_{k+1}}
\]
relates the problem (\ref{TL_k}), (\ref{TL_bc1}) with the piston problem \cite{Holian_Straub_1978, Holian_Flaschka_McLaughlin_1981} for the Toda lattice in the canonical variables:
\[
 \ddot q_k= e^{q_{k-1}-q_k}-e^{q_k-q_{k+1}},\qquad 
 q_0(t)=at,\qquad q_k(0)=\dot q_k(0)=0,\qquad k>0.
\]
This is a model of a system of interacting particles on a semiaxis with initial data corresponding to a state of rest and with a boundary condition at $k=0$ corresponding to a special, infinitely massive particle moving with a constant velocity $a$. Values $a>0$ correspond to shock waves, $a<0$ correspond to rarefaction waves. In the cited papers, it was shown that the values $a=\pm2$ are critical and separate different solution regimes. In particular, it turns out that for $a\le-2$, the solutions of problem (\ref{TL_k}), (\ref{TL_bc1}) asymptotically coincide with the solutions for a lattice with the boundary condition $b_0(t)=0$ instead of $a_0(t)=0$. This corresponds to our main case $\beta=0$, under the substitution (\ref{ab})
\[
 a_k=u_{2k-1}+u_{2k},\qquad b_k=u_{2k}u_{2k+1}
\]
(see Section \ref{s:ab}). Thus, the problem (\ref{VLshock}) with $\beta=0$, corresponding to the self-similar regime, under this substitution maps only to some limiting case of the problem (\ref{TL_k}), (\ref{TL_bc1}). This, however, does not make it trivial, since the substitution (\ref{ab}) is irreversible and, as we have seen, the same solution of the Toda lattice corresponds to different solutions of the Volterra lattice, whose properties differ depending on the value of the parameter $\alpha$.

The case $\beta\ne0$, however, provides a generalization which is different from the problem (\ref{TL_k}), (\ref{TL_bc1}). Indeed, the boundary condition $u_0(t)=\beta\ne0$ yields the relations $b_0=\beta u_1$ and $\dot u_0=0=\beta(u_1-u_{-1})$, that is, $u_{-1}=u_1$; then, after the substitution (\ref{ab}), we obtain $a_0=u_{-1}+u_0=u_1+\beta$. Thus, we arrive at the Toda lattice (\ref{TL_k}) with the initial-boundary conditions
\begin{gather*}
 a_0(t)=\beta+b_0(t)/\beta,\\
 a_1(0)=a_2(0)=\dotsb=\alpha+1/\alpha,\qquad 
 b_0(0)=\alpha\beta,\qquad b_1(0)=b_2(0)=\dotsb=1,
\end{gather*}
containing two parameters, unlike (\ref{TL_bc1}). From the above results of numerical experiments, it follows that in this problem there are also critical values separating different solution regimes.

\section{Generalized Volterra lattices}\label{s:gVL} 

Let us consider an initial-boundary value problem of the form
\begin{gather}
\label{gVL}
 \dot u_n= r(u_n)(u_{n+1}-u_{n-1}),\\
\label{gVLt0}
 u_0(t)=0,\qquad u_1=u_3=\dots=\alpha,\qquad u_2=u_4=\dots=\beta
\end{gather}
where $r(u)$ is a given function such that $r(0)=0$. Numerical experiments show that for some functions, the behavior of the system differs qualitatively little from the VL case $r(u)=u$ discussed above. As an example, we present the graphs in Fig.~\ref{fig:tanh}, corresponding to $r(u)=\tanh u$ (cf. the graphs in Section \ref{s:leading}). The limiting shape and velocity of the decay wave, of course, depend on the choice of $r$. We explain how they are calculated, using non-rigorous reasoning and without further asymptotic analysis.

Substituting the series (\ref{uFG}), (\ref{FGseries}) into (\ref{gVL}), instead of (\ref{fg0sys}), yields the following system with leading coefficients $f=f_0$, $g=g_0$ (for brevity, we omit the subscript, since we do not intend to consider subsequent terms of the series):
\[
 -vxf'(x)=r(f(x))g'(x),\qquad -vxg'(x)=r(g(x))f'(x).
\]
This amounts to the system of equations
\begin{gather}
\label{rfgsys1}
 r(f)r(g)=v^2x^2,\\
\label{rfgsys2}
 r'(f(x))g'(x)+r'(g(x))f'(x)=-2v.
\end{gather}
If $r(u)=u$ then the latter equation is integrated, and we come to the algebraic system (\ref{fg1}) and (\ref{fg2}), from which the functions (\ref{fgc}) were determined. Another case when equation (\ref{rfgsys2}) is integrated is $r(u)=u(\gamma-u)$, that is, the modified Volterra lattice; in this case, the system on $f$ and $g$ takes the form
\[
 f(\gamma-f)g(\gamma-g)=v^2x^2,\qquad \gamma(f+g)-2fg=4c-2vx. 
\]
These equations no longer define a parabola, but a more complicated algebraic curve (also of genus 0), some branches of which serve as attractors for the solution of the lattice equation. However, the rules for selecting these branches require a more detailed analysis. The presence of an additional parameter  $\gamma$ may lead to a greater variety of decay scenarios compared to the Volterra lattice. Therefore, studying the modified lattice remains a more difficult open problem.

It is interesting to note that the cases $r=u$ and $r=u(\gamma-u)$ exhaust the equations of the form (\ref{gVL}) integrable in the sense of the inverse scattering method (according to the Yamilov classification \cite{Yamilov_1984, Yamilov_2006}). In general, the system (\ref{rfgsys1}), (\ref{rfgsys2}) apparently is not solvable by quadratures. However, it can be solved numerically, thereby obtaining the limiting profile of the lattice solution. The system can be resolved with respect to derivatives as follows:
\begin{equation}\label{rfgsys'}
 f'(x)=\frac{2vr(f(x))}{vxr'(f(x))-r(f(x))r'(g(x))},\quad
 g'(x)=\frac{2vr(g(x))}{vxr'(g(x))-r(g(x))r'(f(x))};
\end{equation}
for this system equation (\ref{rfgsys1}) defines the constraint consistent with the dynamics. The value of $v$ is determined by the initial conditions which are specified at the point $x=1$, according to the same rules as in Section \ref{s:leading}, that is,
\begin{equation}\label{fg1'}
\begin{aligned}
 \text{for}~ t\to+\infty:&\quad f(1)=\max(\alpha,\beta),\quad g(1)=\min(\alpha,\beta),\quad
 v=\sqrt{r(\alpha)r(\beta)};\\
 \text{for}~ t\to-\infty:&\quad f(1)=\alpha,\quad g(1)=\beta,\quad v=-\sqrt{r(\alpha)r(\beta)}.
\end{aligned}   
\end{equation}
For the case $t\to+\infty$ and $\alpha=\beta$, a minor technical difficulty arises due to the fact that the denominators in (\ref{rfgsys'}) vanish at $x=1$ (the branches $f(x)$ and $g(x)$ meet at this point and have a vertical tangent). However, even in this case, a numerical solution can be constructed, for example, in a parametric form.

Suppose that the solution $f,g$ of the Cauchy problem (\ref{rfgsys'}), (\ref{fg1'}) exists on the interval $x\in[0,1]$, then these functions determine the desired profile of the lattice solution. If the solution does not extend to the entire interval $x\in[0,1]$, this means that the asymptotics in the form (\ref{uFG}) is not realized for the lattice (\ref{gVL}) with the given function $r(u)$. Determining necessary or sufficient conditions on $r(u)$ under which the hypothesis of asymptotically self-similar behavior is true is beyond the scope of our paper. However, it is clear that there are quite a few such functions, meaning that the self-similar regime is quite universal and deserves further study.

\begin{figure}[t!]
\centerline{\includegraphics[width=\textwidth]{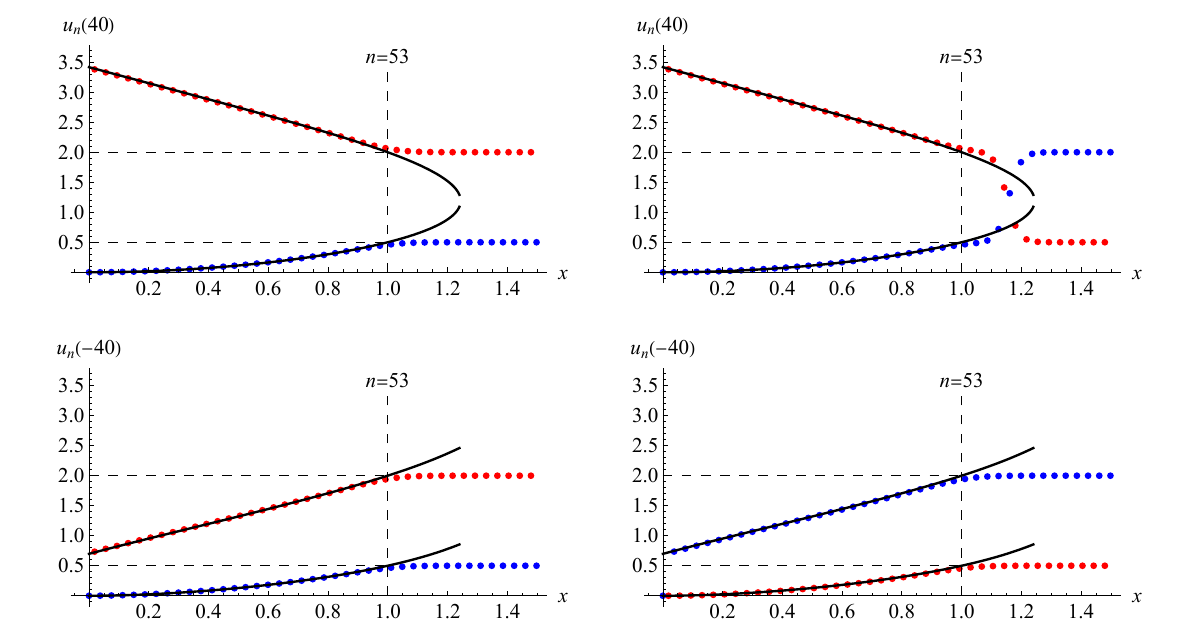}}
\caption{Solution of the problem (\ref{gVL}), (\ref{gVLt0}) for $r(u)=\tanh(u)$ and $\alpha=2$, $\beta=1/2$ (left); $\alpha=1/2$, $\beta=2$ (right). The curves correspond to the solution of the system (\ref{rfgsys'}).}
\label{fig:tanh}
\end{figure}

Fig.~\ref{fig:tanh} compares, for $r(u)=\tanh u$, numerical solutions of the Cauchy problems posed for the lattice equations themselves and for the system on $f$ and $g$. The values $\alpha=2$ and $\beta=1/2$ (and vice versa) are chosen, with $|v|=\sqrt{\tanh(1/2)\tanh(2)}\approx0.667$. The vertical line on the plots with abscissa $x=1$ corresponds to $n=2vt$, and, as we see, it is near this point that the solution departs from the predicted limit curves. The values $f(0)$ and $g(0)$ (found numerically by solving the Cauchy problem (\ref{rfgsys'}), (\ref{fg1'})) determine the post-decay steady state $a,0,a,0,\dots$ to which the system tends: $a\approx3.417$ at $t\to+\infty$ and $a\approx0.698$ at $t\to-\infty$.

\subsection*{Acknowledgments}

The authors are grateful to V.Yu.~Novokshenov for providing the program for calculating the Hastings--McLeod solution. The research of V.E.~Adler is supported by state contract FFWR-2024-0012 of the Landau Institute for Theoretical Physics of the Russian Academy of Sciences. The research of B.I.~Suleimanov is supported by state contract FMRS-2025-0012 of the Institute of Mathematics with Computing Centre---Subdivision of the Ufa Federal Research Centre of Russian Academy of Science.


\end{document}